\documentclass[11pt]{article}
\usepackage{amssymb,amsmath}

\evensidemargin\oddsidemargin
\newtheorem{theorem}{Theorem}
\newtheorem{lemma}{Lemma}
\newtheorem{proposition}{Proposition}
\newtheorem{definition}{Definition}
\newtheorem{corollary}[theorem]{Corollary}
\newtheorem{remark}{Remark}
\newenvironment{proof}[1][Proof]{\textbf{#1.} }{\ \hfill $\Box$ \vspace{0.5cm}}

\begin{document}

\begin{center}
\vspace{1cm}

{\Large{\bf{Canonical transformations for fermions\\[5pt]
in superanalysis}}}

\vspace{0.5cm}

{Joachim Kupsch}

\vspace{0.2cm}

Fachbereich Physik, TU Kaiserslautern,\\[0pt]
D-67653 Kaiserslautern, Germany \\
e-mail: kupsch@physik.uni-kl.de
\end{center}

\bigskip

\begin{abstract}
Canonical transformations (Bogoliubov
transformations) for fermions with an infinite number of degrees of freedom
are studied within a calculus of superanalysis. A continuous representation
of the orthogonal group is constructed on a Grassmann module extension of
the Fock space. The pull-back of these operators to the Fock space yields a
unitary ray representation of the group that implements the Bogoliubov
transformations.
\end{abstract}

\section{Introduction}

Canonical transformations for fermions have been introduced by Bogoliubov
\cite{Bog:1958, BTS:1958} and by Valatin \cite{Valatin:1958, Valatin:1961}
to diagonalize Hamiltonians of the theory of superconductivity. Canonical
transformations for systems with an infinite number of degrees of freedom
have been studied for bosons and for fermions in the book of Friedrichs \cite%
{Friedrichs:1953}. Mathematically minded investigations for fermionic
systems are often based on a study of the Clifford algebra of the field
operators, cf. e.g. \cite{Shale/Stinespring:1965, Araki:1968, Araki:1970,
BMV:1968, BSZ:1992}. The group theoretical structure of the canonical
transformations for fermions is that of the (infinite dimensional)
orthogonal group, which acts on the real Hilbert space that underlies the
complex one particle Hilbert space. An alternative approach to canonical
transformations is therefore the construction of a unitary representation of
this orthogonal group. After partial solutions e.g. in the books \cite%
{Friedrichs:1953, Berezin:1966} a complete construction has been given by
Ruijsenaars \cite{Ruijsenaars:1977, Ruijsenaars:1978} with rigorous normal
ordering expansions.

In this paper we present fermionic canonical transformations with the
methods of infinite dimensional superanalysis as presented in Ref. \cite%
{Kupsch/Smolyanov:1998}. This approach of superanalysis uses Grassmann
modules with a Hilbert norm in contrast to the standard literature, which
either concentrates on the algebraic structure of finite dimensional
superanalysis and supermanifolds \cite{Berezin:1987, DeWitt:1992}, or --
following \cite{Rogers:1980} -- uses a Banach norm for the superalgebra, cf.
e. g. \cite{JP:1981, Khrennikov:1999, Rogers:2007}. The aim of the paper is
twofold. In the first part we recapitulate and amend the genuine infinite
dimensional superanalysis of Ref. \cite{Kupsch/Smolyanov:1998}. The main
tool is the Grassmann module extension of the fermionic Fock space. In the
second part of the paper a representation of the orthogonal group is
constructed on the linear span of fermionic coherent vectors, which exist in
the module Fock space. Then the pull-back of the module operators to the
physical Fock space leads to a unitary ray representation of the orthogonal
group.

The plan of the paper is as follows. In Sec. \ref{Fock} first some facts
about Hilbert and Fock spaces are recapitulated. Then an essentially
self-contained presentation of superanalysis in infinite dimensional spaces
follows. Superanalysis allows to define coherent vectors and Weyl operators
also for fermions. Weyl operators and their interplay with canonical
transformations on the fermionic Fock space are discussed in Sec. \ref{Weyl}%
. The properties of the infinite dimensional orthogonal group are reviewed
in Sec. \ref{orthog}. The construction of the representation of the
orthogonal group on the module Fock space and the pull-back to the physical
Fock space is given in Sec. \ref{rep}. The subclass of orthogonal
transformations, for which the transformed vacuum has still an overlap with
the old vacuum, is investigated in Sec. \ref{invertible}. The representation
of these transformations is given with the methods of superanalysis on the
linear span of coherent vectors. These calculations are the fermionic
counterpart to the representation of the bosonic canonical transformations
in Ref. \cite{Kupsch/Banerjee:2006}. The representation of the full
orthogonal group follows in Sec. \ref{nonInv}. The orbit of the vacuum
generated by all canonical transformations is given in Sec. \ref{vacuum}.
Some proofs and detailed calculations are postponed to the Appendices \ref%
{calc} and \ref{calcRep}.

\section{Fock spaces and superanalysis\label{Fock}}

\subsection{The Fock space of antisymmetric tensors\label{hilbert}}

In this Section we recapitulate some basic statements about Hilbert spaces
and Fock spaces of antisymmetric tensors. Let $\mathcal{H}$ be a complex
separable Hilbert space with inner product $\left( f\mid g\right) $ and with
an antiunitary involution $f\rightarrow
f^{\ast},\,f^{\ast\ast}\equiv(f^{\ast })^{\ast}=f$. Then $\left\langle
f\parallel g\right\rangle :=\left( f^{\ast }\mid g\right) \in\mathbb{C}$ is
a symmetric $\mathbb{C}$-bilinear form $\left\langle f\parallel
g\right\rangle =\left\langle g\parallel f\right\rangle ,\,f,g\in\mathcal{H}$%
. The underlying real Hilbert space of $\mathcal{H}$ is denoted as $\mathcal{%
H}_{\mathbb{R}}$. This space has the inner product $\left( f\mid g\right) _{%
\mathbb{R}}=\mathrm{Re}\,\left( f\mid g\right) $.

We use the following notations for linear operators. The space of all
bounded operators $A$ with operator norm $\left\Vert A\right\Vert $ is $%
\mathcal{L}(\mathcal{H})$. The adjoint operator is denoted by $A^{\dagger}$.
The complex conjugate operator $\bar{A}$ and the transposed operator $A^{T} $
are defined by the identities $\bar{A}f=\left( Af^{\ast}\right) ^{\ast}$ and
$A^{T}f=\left( A^{\dagger}f^{\ast}\right) ^{\ast}$ for all $f\in\mathcal{H}$%
. The usual relations $A^{\dagger}=\left( \bar{A}\right) ^{T}=\overline{%
\left( A^{T}\right) }$ are valid. An operator $A\in \mathcal{L}(\mathcal{H})$%
, which has the property $A^{T}=\pm A$, satisfies the symmetry relation $%
\left\langle f\parallel Ag\right\rangle =\pm\left\langle Af\parallel
g\right\rangle $ for all $f,g\in\mathcal{H}$. It is called
transposition-symmetric or skew symmetric, respectively. The space of all
Hilbert-Schmidt (HS) operators $A$ with norm $\left\Vert A\right\Vert _{2}=%
\sqrt{\mathrm{tr}_{\mathcal{H}}A^{\dagger}A}$ is $\mathcal{L}_{2}(\mathcal{H}%
)$, the space of all trace class or nuclear operators $A$ with norm $%
\left\Vert A\right\Vert _{1}=\mathrm{tr}_{\mathcal{H}}\sqrt{A^{\dagger }A}$
is $\mathcal{L}_{1}(\mathcal{H})$. The HS operators with the property $%
A^{T}=\pm A$ form a closed subspace within $\mathcal{L}_{2}(\mathcal{H})$
for which the notation $\mathcal{L}_{2}^{\pm}(\mathcal{H})$ is used. The
space of all unitary operators in $\mathcal{L}(\mathcal{H})$ is called $%
\mathcal{U}(\mathcal{H}).$ Projection operator always means an orthogonal
projection.

The antisymmetric tensor product or exterior product is written with the
symbol $\wedge $ . The linear span of all tensors $f_{1}\wedge \cdot \cdot
\cdot \wedge f_{n},\,f_{j}\in \mathcal{H},\,j=1,...,n$, is denoted as $%
\mathcal{H}^{\wedge n}$. The space $\mathcal{H}^{\wedge 0}$ is the one
dimensional space $\mathbb{C}$. The Hilbert norm of the space $\mathcal{H}%
^{\wedge n}$ is written as $\left\Vert .\right\Vert _{n}$ and the exterior
product of the vectors $f_{j}\in \mathcal{H},\,j=1,...,n$, is normalized to $%
\left\Vert f_{1}\wedge \cdot \cdot \cdot \wedge f_{n}\right\Vert
_{n}^{2}=\det \left( f_{i}\mid f_{j}\right) $. The completion of the space $%
\mathcal{H}^{\wedge n}$ with the norm $\left\Vert .\right\Vert _{n}$ is the
Hilbert space $\mathcal{A}_{n}(\mathcal{H})$. The exterior product extends
by linearity to the linear space of tensors of finite degree $\mathcal{A}%
_{fin}(\mathcal{H})=\cup _{N=0}^{\infty }\oplus _{n=0}^{N}\mathcal{A}_{n}(%
\mathcal{H})$. This space is an (infinite dimensional) Grassmann algebra.
The unit is the normalized basis vector $1_{vac}$ (vacuum) of the space $%
\mathcal{H}^{\wedge 0}=\mathbb{C}$. If $\mathcal{F}\subset \mathcal{H}$ is a
closed subspace of $\mathcal{H}$ then $\mathcal{A}_{n}(\mathcal{F})$ is the
completed linear span of all tensors $f_{1}\wedge \cdot \cdot \cdot \wedge
f_{n},\,f_{j}\in \mathcal{F}$, and $\mathcal{A}_{fin}(\mathcal{F}):=\cup
_{N=0}^{\infty }\oplus _{n=0}^{N}\mathcal{A}_{n}(\mathcal{F})$ is a
subalgebra of $\mathcal{A}_{fin}(\mathcal{H})$.

Any element $F\in\mathcal{A}_{fin}(\mathcal{H})$ can be decomposed as $%
F=\sum_{n=0}^{\infty}F_{\,n},\,F_{n}\in\mathcal{A}_{n}(\mathcal{H})$, and it
is given the norm
\begin{equation}
\left\Vert F\right\Vert ^{2}=\sum_{n=0}^{\infty}\left\Vert F_{n}\right\Vert
_{n}^{2}.  \label{h3}
\end{equation}
The completion of $\mathcal{A}_{fin}(\mathcal{H})$ with this norm is the
standard Fock space of antisymmetric tensors $\mathcal{A}(\mathcal{H}%
)=\oplus_{n=0}^{\infty}\mathcal{A}_{n}(\mathcal{H}).$ The inner product of
two elements $F,G\in$ $\mathcal{A}(\mathcal{H})$ is written as $\left( F\mid
G\right) $. The antiunitary involution $f\rightarrow f^{\ast}$ on $\mathcal{H%
}$ can be extended uniquely to an antiunitary involution $F\rightarrow
F^{\ast}$on $\mathcal{A}(\mathcal{H})$ with the rule $\left( F\wedge
G\right) ^{\ast}=G^{\ast}\wedge F^{\ast}$. Then $\left\langle F\parallel
G\right\rangle :=\left( F^{\ast}\mid G\right) $ is a $\mathbb{C}$-bilinear
symmetric form on $\mathcal{A}(\mathcal{H})$. Any bounded operator $B\in%
\mathcal{L}(\mathcal{H})$ can be lifted to an operator $\Gamma(B)$ on $%
\mathcal{A}(\mathcal{H})$ by the rules $\Gamma(B)1_{vac}=1_{vac}$ and $%
\Gamma(B)\left( f_{1}\wedge\cdot\cdot\cdot\wedge f_{n}\right)
:=Bf_{1}\wedge\cdot\cdot\cdot\wedge Bf_{n},~f_{j}\in\mathcal{H},\,n\in
\mathbb{N}$. The operator $\Gamma(B)$ is a contraction, if $B$ is a
contraction, and it is isometric/unitary, if $B$ is isometric/unitary.

For products of vectors we use the following notation. Let $\mathbf{A\subset
}\mathbb{N}$ be a finite subset of the natural numbers with cardinality $%
\left\vert \mathbf{A}\right\vert =n\geq1$. Then the elements of $\mathbf{A}$
can be ordered $\mathbf{A}=\left\{ a_{1}<...<a_{n}\right\} $. Given the
vectors $f_{a},\,a\in\mathbf{A}$ the tensor $f_{\mathbf{A}}$ is defined as
the product $f_{a_{1}}\wedge...\wedge f_{a_{n}}$ with ordered indices. For
the empty set $\mathbf{A}=\emptyset$ we define $f_{\emptyset}=1_{vac}$. If $%
f_{b},\,b\in\mathbf{B}$, is another family of vectors, indexed by the finite
set $\mathbf{B\subset}\mathbb{N},\,\mathbf{A}\cap\mathbf{B}=\emptyset$, then
the exterior product of these vectors is $f_{\mathbf{A}}\wedge f_{\mathbf{B}%
}=(-1)^{\tau(\mathbf{A},\mathbf{B})}f_{\mathbf{A}\cup\mathbf{B}}$. The
exponent $\tau(\mathbf{A},\mathbf{B}):=\#\left\{ (a,b)\in\mathbf{A}\times%
\mathbf{B}\mid a>b\right\} $ counts the number of inversions. In the same
notation we write $z_{\mathbf{A}}$ for the product $z_{a_{1}}z_{a_{2}}\cdot%
\cdot\cdot z_{a_{n}}$ of the complex numbers $z_{a},\,a\in\mathbf{A}$, with $%
z_{\emptyset}=1$. The symbol $\sum_{\mathbf{A}\subset\mathbb{N}}$ always
means summation over the power set $\mathcal{P}(\mathbb{N})$ of $\mathbb{N}$%
, i.e. over all \textit{finite} subsets of $\mathbb{N}$ including the empty
set.

\subsection{Superanalysis and coherent vectors\label{super}}

Superanalysis allows to define coherent vectors and Weyl operators also for
fermions. An approach of superanalysis for spaces with infinite dimensions
has been developed in \cite{Kupsch/Smolyanov:1998}. In this Section we
present a slightly modified and extended version of superanalysis, which is
used in the subsequent calculations. Some proofs are given in the Appendix %
\ref{calc} or in Ref. \cite{Kupsch/Smolyanov:1998}. For the superalgebra $%
\Lambda $ we choose a Grassmann algebra $\Lambda $ which differs by
notations and a weaker Hilbert norm from the Fock space $\mathcal{A}(%
\mathcal{H})$. The Grassmann algebra is the direct sum $\Lambda =\oplus
_{p\geq 0}\Lambda _{p}$ of the subspaces $\Lambda _{p}$ of tensors of degree
$p$. The generating space is the infinite dimensional Hilbert space $\Lambda
_{1}$. The space $\Lambda _{p}$ is given the Hilbert norm $\left\Vert
.\right\Vert _{p}$ of antisymmetric tensors of degree $p$ normalized to the
determinant (as for $\mathcal{A}_{p}(\mathcal{H})$). The unit is denoted by $%
\kappa _{0}$, and it has the norm $\left\Vert \kappa _{0}\right\Vert _{0}=1$%
. The topology of $\Lambda $ is then defined by the Hilbert norm%
\begin{equation}
\left\Vert \lambda \right\Vert _{\Lambda }^{2}=\sum_{p=0}^{\infty
}(p!)^{-2}\left\Vert \lambda _{p}\right\Vert _{p}^{2}  \label{h4}
\end{equation}%
if $\lambda =\sum_{p=0}^{\infty }\lambda _{p},\,\lambda _{p}\in \Lambda _{p}$%
. The antisymmetric tensor product of two elements $\lambda _{1}$ and $%
\lambda _{2}$ of $\Lambda $ is now denoted as Grassmann product and it is
written as $\lambda _{1}\lambda _{2}$. As consequence of the topology (\ref%
{h4}) this product is continuous with the estimate $\left\Vert \lambda
_{1}\lambda _{2}\right\Vert _{\Lambda }\leq \sqrt{3}\left\Vert \lambda
_{1}\right\Vert _{\Lambda }\left\Vert \lambda _{2}\right\Vert _{\Lambda }$,
cf. \cite{Kupsch/Smolyanov:1998} Appendix A or \cite{Kupsch/Smolyanov:2000}.
If $\lambda _{1}\in \Lambda _{1}$ and $\lambda _{2}\in \Lambda $ the
stronger estimate $\left\Vert \lambda _{1}\lambda _{2}\right\Vert _{\Lambda
}\leq \left\Vert \lambda _{1}\right\Vert _{\Lambda }\left\Vert \lambda
_{2}\right\Vert _{\Lambda }$ follows as for the standard Fock space. This
Grassmann algebra is a superalgebra $\Lambda =\Lambda _{\bar{0}}\oplus
\Lambda _{\bar{1}}$ with the even part $\Lambda _{\bar{0}}$ and the odd part
$\Lambda _{\bar{1}}$ consisting of tensors of even or odd degree,
respectively. As additional structure we introduce an antiunitary involution
$\kappa \rightarrow \kappa ^{\ast }$ in the generating space $\Lambda _{1}$.
This involution is extended to an antiunitary involution $\lambda
\rightarrow \lambda ^{\ast }\ $on $\Lambda $ by the usual rules.

\begin{remark}
The Grassmann algebra $\Lambda$ is needed for a correct bookkeeping of the
fermionic degrees of freedom. There is no canonical way for such
constructions, and instead of a Grassmann algebra one can take a more
general superalgebra $\Lambda$. In the literature about superanalysis the
topology of $\Lambda$ is either not discussed, or -- following \cite%
{Rogers:1980} -- the Grassmann algebra is equipped with a Banach space
topology, cf. e.g. \cite{JP:1981, Khrennikov:1999, Rogers:2007}. But a
Hilbert space topology is easier to handle, and the pull-back to the Fock
space of standard physics is more transparent.
\end{remark}

The Hilbert space $\mathcal{H}$, the algebra of antisymmetric tensors $%
\mathcal{A}_{fin}(\mathcal{H})$ and the Fock space $\mathcal{A}(\mathcal{H})$
can be extended to the $\Lambda $-modules $\mathcal{H}^{\Lambda }=\Lambda
\widehat{\otimes }\mathcal{H},\,\mathcal{A}_{fin}^{\Lambda }(\mathcal{H}%
)=\cup _{N=0}^{\infty }\oplus _{n=0}^{N}\Lambda \widehat{\otimes }\mathcal{A}%
_{n}(\mathcal{H})$ and $\mathcal{A}^{\Lambda }(\mathcal{H})=\Lambda \widehat{%
\otimes }\mathcal{A}(\mathcal{H})$.\footnote{%
The tensor product $\mathcal{H}_{1}\otimes \mathcal{H}_{2}$ of two Hilbert
spaces means the algebraic tensor space. The completion of $\mathcal{H}%
_{1}\otimes \mathcal{H}_{2}$ with the Hilbert cross norm is $\mathcal{H}_{1}%
\widehat{\otimes }\mathcal{H}_{2}$.} The Hilbert norm of the module Fock
space $\mathcal{A}^{\Lambda }(\mathcal{H})$ is denoted as $\left\Vert
.\right\Vert _{\otimes }$. The $\Lambda $-linear space $\mathcal{A}%
_{fin}^{\Lambda }(\mathcal{H})$ is a pre-Hilbert space with completion $%
\mathcal{A}^{\Lambda }(\mathcal{H})$. The tensor product of $\mathcal{A}%
_{fin}(\mathcal{H})$ has a $\Lambda $-linear extension $\mathcal{A}%
_{fin}^{\Lambda }(\mathcal{H})\ni \Xi _{1},\Xi _{2}\rightarrow \Xi _{1}\circ
\Xi _{2}\in \mathcal{A}_{fin}^{\Lambda }(\mathcal{H})$. This product is
uniquely determined by the rule $\Xi _{1}\circ \Xi _{2}=\lambda _{1}\lambda
_{2}\otimes (F_{1}\wedge F_{2})$ if $\Xi _{j}=\lambda _{j}\otimes F_{j}\in
\Lambda \otimes \mathcal{A}_{fin}(\mathcal{H)}\,,$ $j=1,2$. Any tensor $\Xi $
can be written as series $\Xi =\sum_{p=0,n=0}^{\infty }\Xi _{p,n}$ with $\Xi
_{p,n}\in \Lambda _{p}\widehat{\otimes }\mathcal{A}_{n}(\mathcal{H}).$ If
the set $\left\{ p+n\mid \Xi _{p,n}\neq 0\right\} $ contains only even (odd)
numbers, we say the tensor $\Xi $ has even (odd) parity, $\pi (\Xi
)=0\,(\,1\,)$. With this parity the module $\mathcal{A}_{fin}^{\Lambda }(%
\mathcal{H})$ is a $\mathbb{Z}_{2}$-graded algebra. If the tensors $\Theta $
and $\Xi $ have the parities $\pi (\Theta )=p$ and $\pi (\Xi )=q$, the
product satisfies $\Theta \circ \Xi =(-1)^{pq}\,\Xi \circ \Theta $. The
product $\Theta \circ \Xi $ is not defined on the full the module Fock space
$\mathcal{A}^{\Lambda }(\mathcal{H})$, but it can be extended to a larger
class than $\mathcal{A}_{fin}^{\Lambda }(\mathcal{H)}$ by continuity
arguments, cf. Appendix \ref{prod} and Lemma \ref{product2}.

The inner product of $\mathcal{A}(\mathcal{H})$ has a unique module
extension to $\mathcal{A}^{\Lambda }(\mathcal{H})$ with the property $\left(
\lambda \otimes F\mid \mu \otimes G\right) =\lambda ^{\ast }\mu \left( F\mid
G\right) \in \Lambda $ where $\lambda ,\mu \in \Lambda $ and $F,G\in
\mathcal{A}(\mathcal{H})$. This $\Lambda $-valued inner product $\mathcal{A}%
^{\Lambda }(\mathcal{H})\times \mathcal{A}^{\Lambda }(\mathcal{H})\ni \left(
\Theta ,\Xi \right) \rightarrow \left( \Theta \mid \Xi \right) \in \Lambda $
has the hermiticity property $\left( \Theta \mid \Xi \right) ^{\ast }=\left(
\Xi \mid \Theta \right) $, and it satisfies the norm estimate $\left\Vert
\left( \Theta \mid \Xi \right) \right\Vert _{\Lambda }\leq \sqrt{3}%
\left\Vert \Theta \right\Vert _{\otimes }\left\Vert \Xi \right\Vert
_{\otimes }$, cf. \cite{Kupsch/Smolyanov:1998}. If the tensors $\Theta $ and
$\Xi $ have the same parity, $\pi (\Theta )=\pi (\Xi )$, the form $\,\left(
\Theta \mid \Xi \right) \in \Lambda _{\bar{0}}$ is an even element of $%
\Lambda $. The involution of $\mathcal{A}(\mathcal{H})$ can be extended to
an involution $\Xi \rightarrow \Xi ^{\ast }$ on $\mathcal{A}^{\Lambda }(%
\mathcal{H})$ with the property $\left( \lambda \otimes F\right) ^{\ast
}=\lambda ^{\ast }\otimes F^{\ast }$, and the bilinear form $\left\langle
.\parallel .\right\rangle $ of $\mathcal{A}(\mathcal{H})$ has a unique $%
\Lambda $-extension $\mathcal{A}^{\Lambda }(\mathcal{H})\times \mathcal{A}%
^{\Lambda }(\mathcal{H})\ni \left( \Theta ,\Xi \right) \rightarrow
\left\langle \Theta \parallel \Xi \right\rangle :=\left( \Theta ^{\ast }\mid
\Xi \right) \in \Lambda $.

The factorization property of the inner product of $\mathcal{A}(\mathcal{H})$
for tensors in orthogonal subspaces leads to the following factorization of $%
\left( \Theta \mid \Xi \right) $:

\begin{lemma}
Let $\mathcal{H}_{1}$ and $\mathcal{H}_{1}$ be two orthogonal closed
subspaces of the Hilbert space $\mathcal{H}$. The tensors $\Theta _{j},\,\Xi
_{j}\in \mathcal{A}^{\Lambda }(\mathcal{H})\,,\,j=1,2$, are restricted by :%
\newline
(i) The tensors are elements of the orthogonal subspaces, $\Theta _{1},\,\Xi
_{1}\in \mathcal{A}^{\Lambda }(\mathcal{H}_{1})$ and

$\Theta _{2},\,\Xi _{2}\in \mathcal{A}^{\Lambda }(\mathcal{H}_{2})$. \newline
(ii) The tensors $\Theta _{1}$ and $\Xi _{1}$ have equal parity.\newline
Then the products $\Theta _{1}\circ \Theta _{2}$ and $\Xi _{1}\circ \Xi _{2}$
are defined and the factorization $\left( \Theta _{1}\circ \Theta _{2}\mid
\Xi _{1}\circ \Xi _{2}\right) =\left( \Theta _{1}\mid \Xi _{1}\right) \left(
\Theta _{2}\mid \Xi _{2}\right) $ is true.
\end{lemma}

If $\mu\in\Lambda$ and $T\in\mathcal{L}\left( \mathcal{A}(\mathcal{H}%
)\right) $ then there is a unique bounded operator $\mu\otimes T$ on $%
\mathcal{A}^{\Lambda}(\mathcal{H})$ that maps $\lambda\otimes F\in
\Lambda\otimes\mathcal{A}(\mathcal{H})$ onto $\mu\lambda\otimes T\,F\in
\Lambda\otimes\mathcal{A}(\mathcal{H})$.

\begin{definition}
\label{regular} We call a continuous $\mathbb{C}$-linear operator $\hat{T}$
on $\mathcal{A}^{\Lambda}(\mathcal{H})=\Lambda\widehat{\otimes}\mathcal{A}(%
\mathcal{H})$ a \textit{regular operator}, if it is can be represented as $%
\hat{T}=\sum_{j\in\mathbf{J}}\mu_{j}\otimes T_{j}$ with $\mu_{j}\in
\Lambda,\,T_{j}\in\mathcal{L}\left( \mathcal{A}(\mathcal{H})\right) $ and a
finite or countable index set $\mathbf{J}\subset\mathbb{N}$. The series $%
\sum_{j\in\mathbf{J}}\left\Vert \mu_{j}\right\Vert _{\Lambda}\left\Vert
T_{j}\right\Vert $ has to converge.
\end{definition}

\noindent The operator norm of $\hat{T}=\sum_{j\in\mathbf{J}}\mu_{j}\otimes
T_{j}$ has the upper bound $\left\Vert \hat{T}\right\Vert \leq\sqrt{3}%
\sum_{j\in\mathbf{J}}\left\Vert \mu_{j}\right\Vert _{\Lambda}\left\Vert
T_{j}\right\Vert $. If $\hat{T}_{1}$ and $\hat{T}_{2}$ are regular
operators, then the product $\hat{T}_{1}\hat{T}_{2}$ is a regular operator.
Let $\hat{T}$ be a bounded operator on $\mathcal{A}^{\Lambda}(\mathcal{H})$,
then the operator $\hat{T}^{+}$ is called the superadjoint operator of $\hat{%
T}$, if the identity $\left( \Theta\mid\hat{T}\Xi\right) =\left( \hat{T}%
^{+}\Theta\mid\Xi\right) $ is valid for all $\Theta,\Xi\in\mathcal{A}%
^{\Lambda }(\mathcal{H})$. For a regular operator $\hat{T}%
=\sum_{j}\lambda_{j}\otimes T_{j}$ the superadjoint operator is the regular
operator $\hat{T}^{+}=\sum _{j}\lambda_{j}^{\ast}\otimes T_{j}^{\dagger}$.
In general the superadjoint operator is not the adjoint operator in the
standard definition using the $\mathbb{C}$-valued inner product of the
Hilbert space $\mathcal{A}^{\Lambda }(\mathcal{H})$.

For the subsequent constructions we define the fermionic \textit{superspace }%
as the completed tensor space $\mathcal{H}_{\Lambda }:=\Lambda _{1}\widehat{%
\otimes }\mathcal{H}$. The algebraic superspace $\mathcal{H}_{\Lambda
}^{_{alg}}:=\Lambda _{1}\otimes \mathcal{H}$ is dense in $\mathcal{H}%
_{\Lambda }$. The inclusions $\mathcal{H}_{\Lambda }^{_{alg}}\subset
\mathcal{H}_{\Lambda }\subset \mathcal{H}^{\Lambda }\subset \mathcal{A}%
^{\Lambda }(\mathcal{H})$ are obvious. The elements of $\mathcal{H}_{\Lambda
}$ have even parity, and for $\xi ,\eta \in \mathcal{H}_{\Lambda }$ the $%
\Lambda $-extended inner product has the hermiticity properties $\left( \xi
\mid \eta \right) ^{\ast }=\left( \eta \mid \xi \right) =-\left( \xi ^{\ast
}\mid \eta ^{\ast }\right) $.

\begin{remark}
\label{Smol}In Ref. \cite{Kupsch/Smolyanov:1998} the fermionic part of the
superspace is introduced as $\Lambda_{\bar{1}}\widehat{\otimes}\mathcal{H}$
with the full odd subspace $\Lambda_{\bar{1}}$ of the Grassmann algebra $%
\Lambda$, and the coherent vectors $\exp\,\xi$ are defined with arguments $%
\xi\in\Lambda_{\bar{1}}\widehat{\otimes}\mathcal{H}$. Here we use the
strictly smaller superspace $\mathcal{H}_{\Lambda}=\Lambda_{1}\widehat{%
\otimes }\mathcal{H}$ with the generating Hilbert space $\Lambda_{1}$ of $%
\Lambda$. In \cite{Kupsch/Smolyanov:1998} this space was called \textit{%
restricted superspace}. Calculations with $\Lambda_{1}\widehat{\otimes}%
\mathcal{H}$ allow better norm estimates.
\end{remark}

In choosing the superspace $\mathcal{H}_{\Lambda}=\Lambda_{1}\widehat{%
\otimes }\mathcal{H}$ the multiplication with supervectors $\xi\in\mathcal{H}%
_{\Lambda}$ has the following property:

\begin{lemma}
\label{product2}Let $\xi\in\mathcal{H}_{\Lambda}=\Lambda_{1}\widehat{\otimes
}\mathcal{H}$ be a supervector, then the mapping $\mathcal{A}_{fin}^{\Lambda
}(\mathcal{H})\ni\Theta\rightarrow\xi\circ\Theta\in\mathcal{A}%
_{fin}^{\Lambda }(\mathcal{H})$ can be extended to a continuous mapping on $%
\mathcal{A}^{\Lambda}(\mathcal{H})$ with the norm estimate%
\begin{equation}
\left\Vert \xi\circ\Theta\right\Vert _{\otimes}\leq\left\Vert \xi\right\Vert
_{\otimes}\left\Vert \Theta\right\Vert _{\otimes},\;\Theta\in\mathcal{A}%
^{\Lambda}(\mathcal{H}).  \label{h5}
\end{equation}
\end{lemma}

The proof of this Lemma is given in the Appendix \ref{prod}.

If $A$ is a bounded operator on $\mathcal{H}$, then $\kappa _{0}\otimes A$
is a bounded operator on $\mathcal{H}_{\Lambda }$ that maps $\xi =\mu
\otimes f\in \mathcal{H}_{\Lambda }$ onto $\left( \kappa _{0}\otimes
A\right) \xi =\mu \otimes Af\in \mathcal{H}_{\Lambda }$. To simplify the
notations we write $A\,\xi $ instead of $\left( \kappa _{0}\otimes A\right)
\xi $.

The fermionic coherent vectors are defined as \cite{Kupsch/Smolyanov:1998}%
\begin{equation}
\mathcal{H}_{\Lambda }\ni \xi \rightarrow \exp \,\xi =\sum_{p=0}^{\infty
}\left( p!\right) ^{-1}\xi ^{p}\in \mathcal{A}^{\Lambda }(\mathcal{H}).
\label{h6}
\end{equation}%
The norm convergence of this series follows from the estimate (\ref{h5})
with $\left\Vert \exp \,\xi \right\Vert _{\otimes }^{2}\leq \newline
\sum_{p}\left( p!\right) ^{-2}\left\Vert \xi ^{p}\right\Vert _{\otimes
}^{2}\leq \sum_{p}\left( p!\right) ^{-2}\left\Vert \xi \right\Vert _{\otimes
}^{2p}$. The exponential $\exp \,\xi $ is therefore an entire analytic
function on $\mathcal{H}_{\Lambda }$, and the usual multiplication rule $%
\left( \exp \,\xi \right) \circ \left( \exp \,\eta \right) =\exp \,(\xi
+\eta )$ holds. The exponential vectors have even parity. If $\left\{
e_{k}\mid k\in \mathbb{N}\right\} $ is an ON basis of $\mathcal{H}$, the
tensors $\left\{ e_{\mathbf{K}}\mid \mathbf{K}\subset \mathbb{N}\right\} $
are an ON basis of $\mathcal{A}(\mathcal{H})$. Any supervector has the
representation $\xi =\sum_{k\in \mathbb{N}}\lambda _{k}\otimes e_{k}$ with
elements $\lambda _{k}\in \Lambda _{1},\,\sum_{k\in \mathbb{N}}\left\Vert
\lambda _{k}\right\Vert _{\Lambda }^{2}=\left\Vert \xi \right\Vert _{\otimes
}^{2}$. Then the exponential vector is the series $\exp \,\xi =\sum_{\mathbf{%
K}\subset \mathbb{N}}\lambda _{\mathbf{K}}\otimes e_{\mathbf{K}}$, and $%
\left( \exp \,\xi \mid \exp \,\xi \right) =\sum_{\mathbf{K}\subset \mathbb{N}%
}(\lambda _{\mathbf{K}})^{\ast }\lambda _{\mathbf{K}}=\exp \left( \xi \mid
\xi \right) $ follows. The $\Lambda $-inner product of two exponential
vectors is therefore%
\begin{equation}
\left( \exp \,\xi \mid \exp \,\eta \right) =\exp \left( \xi \mid \eta
\right) .  \label{h16}
\end{equation}

If $\Theta$ and $\Xi$ are tensors in $\mathcal{A}^{\Lambda}(\mathcal{H})$,
which have even parity, and for which the product $\Theta\circ\Xi$ is
defined, the formula
\begin{equation}
\left( \exp\,\xi\mid\Theta\circ\Xi\right) =\left( \exp\,\xi\mid
\Theta\right) \left( \exp\,\xi\mid\Xi\right)  \label{h17}
\end{equation}
is valid. This rule is obviously true for $\Theta=\exp\,\eta$ and $\Xi
=\exp\,\zeta$ with $\eta,\zeta\in\mathcal{H}_{\Lambda}$. If $\eta$ is a
tensor of degree $2p$ and $\zeta$ is a tensor of degree $2q$, the identity (%
\ref{h17}) follows for the terms with the highest tensor degrees, i.e. $%
\Theta=\eta^{p}$ and $\Xi=\zeta^{q}$. Some algebra $-$ using techniques of
Appendix B in \cite{Kupsch/Smolyanov:1998} -- leads to the rule for $%
\Theta,\Xi\in\mathcal{A}_{fin}^{\Lambda}(\mathcal{H})$ with even parity.
Then continuity arguments extend the identity to those elements of $\mathcal{%
A}^{\Lambda}(\mathcal{H})$, for which $\Theta\circ\Xi$ is defined.

The $\mathbb{C}$-linear span of $\left\{ \exp\,\xi\mid\xi\in\mathcal{H}%
_{\Lambda}\right\} $ is called $\mathcal{C}(\mathcal{H}_{\Lambda})$, the $%
\Lambda$-linear span $\Lambda\otimes\mathcal{C}(\mathcal{H}_{\Lambda})$ is
called $\mathcal{C}^{\Lambda}(\mathcal{H}_{\Lambda})$. With the
identification $\mathcal{C}(\mathcal{H}_{\Lambda})\simeq\kappa_{0}\otimes%
\mathcal{C}(\mathcal{H}_{\Lambda})$ the space $\mathcal{C}(\mathcal{H}%
_{\Lambda})$ is a subspace of $\mathcal{C}^{\Lambda}(\mathcal{H}_{\Lambda})$%
, and we have the inclusions $\mathcal{H}_{\Lambda}\subset\mathcal{C}(%
\mathcal{H}_{\Lambda })\subset\mathcal{C}^{\Lambda}(\mathcal{H}%
_{\Lambda})\subset\mathcal{A}^{\Lambda}(\mathcal{H})$. Despite the values
taken by functions $\Phi \in\mathcal{C}(\mathcal{H}_{\Lambda})$ only
trivially intersect with the fermionic Fock space $\mathcal{A}(\mathcal{H})$
the following Lemmata are true.

\begin{lemma}
\label{analytic}A tensor $\Xi\in\mathcal{A}^{\Lambda}(\mathcal{H})$ is
uniquely determined by the function \newline
$\mathcal{H}_{\Lambda}\ni \zeta\rightarrow\left( \exp\,\zeta\mid\Xi\right)
\in{\Lambda}$.
\end{lemma}

The proof of this Lemma follows from Sec. 5.3 of \cite{Kupsch/Smolyanov:1998}
Lemma 4 and Corollary 2.

\begin{lemma}
\label{regOp1} Let $\hat{T}$ be a regular operator on $\mathcal{A}^{\Lambda
}(\mathcal{H})$. Then $\hat{T}$ is uniquely determined by its values on the
set of coherent vectors $\left\{ \exp\xi\mid\xi\in\mathcal{H}_{\Lambda
}\right\} $.
\end{lemma}

\begin{proof}
The operator $\hat{T}$ is a linear and continuous operator on the Hilbert
space $\mathcal{A}^{\Lambda}(\mathcal{H})$. It is sufficient to prove that $%
\hat{T}\exp\xi=0,\,\xi\in\mathcal{H}_{\Lambda}$, implies $\hat{T}\,\Xi=0$
for $\Xi\in\mathcal{A}^{\Lambda}(\mathcal{H})$. Given the regular operator $%
\hat{T}=\sum_{j}\mu_{j}\otimes T_{j}$ it has the bounded superadjoint $\hat{T%
}^{+}=\sum_{j}\mu_{j}^{\ast}\otimes T_{j}^{\dagger}$. Assume $\hat {T}%
\exp\xi=0$ is true for all $\xi\in\mathcal{H}_{\Lambda}$. Then $0=\left(
\hat{T}\exp\xi\mid\exp\eta\right) =\left( \exp\xi\mid\hat{T}^{+}\exp
\eta\right) $ holds for all $\xi,\eta\in\mathcal{H}_{\Lambda}$, and -- as a
consequence of Lemma \ref{analytic} -- the identity $\hat{T}^{+}\exp\eta=0$
is true for all $\eta\in\mathcal{H}_{\Lambda}$. But then we have $\left(
\exp \xi\mid\hat{T}\Xi\right) =\left( \hat{T}^{+}\exp\xi\mid\Xi\right) =0$
for$\,\xi\in\mathcal{H}_{\Lambda}$ and $\Xi\in\mathcal{A}^{\Lambda }(%
\mathcal{H})$. Hence $\hat{T}\,\Xi=0$ is true for $\Xi\in\mathcal{A}%
^{\Lambda}(\mathcal{H})$.
\end{proof}

This Lemma implies

\begin{corollary}
\label{regOp2} Let $\hat{T}$ be a regular operator on $\mathcal{A}^{\Lambda
}(\mathcal{H})$. Then $\hat{T}$ is uniquely determined by the function $%
\mathcal{H}_{\Lambda}\ni\xi\rightarrow\left( \exp\xi\mid\hat{T}\,\exp
\xi\right) \in\Lambda$.
\end{corollary}

\begin{proof}
The Lemmata \ref{analytic} and \ref{regOp1} imply that $\hat{T}$ is
determined by the function $\varphi(\xi,\eta):=\left( \exp\xi\mid\hat{T}%
\,\exp \eta\right) $, which is analytic in $\eta\in\mathcal{H}_{\Lambda} $
and antianalytic in $\xi\in\mathcal{H}_{\Lambda}$. Such a function is
uniquely determined by its values on the diagonal $\xi=\eta$.
\end{proof}

The Fock space creation and annihilation operators $a^{+}(h)F=h\wedge
F,\,h\in \mathcal{H},\,F\in \mathcal{A}(\mathcal{H})$, and $a^{-}(h)=\left(
a^{+}(h)\right) ^{\dagger }$ have an extension to $\mathbb{C}$-linear
operators $b^{+}(\eta )$ and $b^{-}(\eta ),\,\eta \in \mathcal{H}_{\Lambda }$%
, on $\mathcal{A}^{\Lambda }(\mathcal{H})$. The operator $b^{+}(\eta )$ is
simply given by $b^{+}(\eta )\,\Xi =\eta \circ \Xi $ with $\,\eta \in
\mathcal{H}_{\Lambda }$ for all $\Xi \in \mathcal{A}^{\Lambda }(\mathcal{H})$%
. The operator $b^{-}(\eta )$ is the superadjoint$\;b^{-}(\eta )=\left(
b^{+}(\eta )\right) ^{+}$. For $\eta =\sum_{j}\mu _{j}\otimes f_{j}\in
\mathcal{H}_{\Lambda }$ we obtain $b^{+}(\eta )=\sum_{j}\mu _{j}\otimes
a^{+}(f_{j})$ and $b^{-}(\eta )=\sum_{j}\mu _{j}^{\ast }\otimes a^{-}(f_{j})$%
. The operators $b^{\pm }(\eta )$ are continuous with operator norm $%
\left\Vert b^{\pm }(\eta )\right\Vert \leq \left\Vert \eta \right\Vert
_{\otimes },\,\eta \in \mathcal{H}_{\Lambda }$, cf. the end of Appendix \ref%
{prod}. The mapping $\mathcal{H}_{\Lambda }\ni \eta \rightarrow b^{+}(\eta )$
is $\mathbb{C}$-linear, and $\eta \rightarrow b^{-}(\eta )$ is $\mathbb{C}$%
-antilinear. The operators $b^{\pm }(\eta )$ are regular operators. Lemma %
\ref{regOp1} implies that the operators $b^{\pm }(\eta )$ are uniquely
determined by their values on coherent vectors $\exp \,\xi ,\,\xi \in
\mathcal{H}_{\Lambda }$. A simple calculation gives for all $\xi ,\eta \in
\mathcal{H}_{\Lambda }$%
\begin{equation}
b^{+}(\eta )\exp \xi =\eta \cdot \exp \xi ,\;b^{-}(\eta )\exp \xi =\left(
\eta \mid \xi \right) \exp \xi .  \label{h10}
\end{equation}

\subsection{Exponentials of tensors of second degree\label{exp2}}

Given a skew symmetric HS operator $X$ on $\mathcal{H}$, then there exists
exactly one tensor $\Omega(X)\in\mathcal{A}_{2}(\mathcal{H})$ such that $%
\left\langle \Omega(X)\parallel f\wedge g\right\rangle =\left\langle
f\parallel Xg\right\rangle $ for all $f,g\in\mathcal{H}$, cf. Appendix \ref%
{exponentials}. The exponentials of tensors in $\mathcal{A}_{2}(\mathcal{H})$
have been investigated in the literature, e.g. in Chap. 12 of \cite%
{Pressley/Segal:1986}. But for completeness and to fix the normalizations we
derive the following statements in Appendix \ref{exponentials}.

\begin{enumerate}
\item For $X\in \mathcal{L}_{2}^{-}(\mathcal{H})$ the exponential series
\begin{equation}
\exp \Omega (X)=1_{vac}+\frac{1}{1!}\Omega (X)+\frac{1}{2!}\Omega (X)\wedge
\Omega (X)+...  \label{h11}
\end{equation}%
converges uniformly in $\mathcal{A}(\mathcal{H})$, and the mapping $\mathcal{%
L}_{2}^{-}(\mathcal{H})\ni X\rightarrow \exp \Omega (X)\in \mathcal{A}(%
\mathcal{H})$ is entire analytic.

\item The inner product of two of such tensors is%
\begin{equation}
\left( \exp \Omega (X)\mid \exp \Omega (Y)\right) =\sqrt{\det (I+X^{\dagger
}Y)}=\sqrt{\det (I+YX^{\dagger })},\;X,Y\in \mathcal{L}_{2}^{-}(\mathcal{H}).
\label{h15}
\end{equation}
\end{enumerate}

The mapping $\mathcal{A}(\mathcal{H})\ni F\rightarrow \kappa _{0}\otimes
F\in \mathcal{A}^{\Lambda }(\mathcal{H})$ gives a natural embedding of the
Fock space $\mathcal{A}(\mathcal{H})$ into the $\Lambda $-module $\mathcal{A}%
^{\Lambda }(\mathcal{H})=\Lambda \widehat{\otimes }\mathcal{A}(\mathcal{H})$%
. The tensors (\ref{h11}) can therefore be taken as the elements $\Psi
(X):=\kappa _{0}\otimes \exp \Omega (X)=\exp \left( \kappa _{0}\otimes
\Omega (X)\right) $ of $\mathcal{A}^{\Lambda }(\mathcal{H})$. In Appendix %
\ref{ultra} we prove that the products
\begin{equation}
\Psi (X,\xi ):=\left( \exp \xi \right) \circ \Psi (X)=\Psi (X)\circ \left(
\exp \xi \right) =\exp \left( \xi +\kappa _{0}\otimes \Omega (X)\right)
\label{h12}
\end{equation}%
of the coherent vectors (\ref{h6}) and of the exponentials $\Psi (X)$ are
well defined elements of $\mathcal{A}^{\Lambda }(\mathcal{H})$ for all $X\in
\mathcal{L}_{2}^{-}(\mathcal{H})$ and $\xi \in \mathcal{H}_{\Lambda }$ The
mapping $\mathcal{L}_{2}^{-}(\mathcal{H})\times \mathcal{H}_{\Lambda }\ni
(X,\xi )\rightarrow \Psi (X,\xi )\in \mathcal{A}^{\Lambda }(\mathcal{H})$ is
entire analytic. The vectors $\Psi (X,\xi )$ are called \textit{%
ultracoherent vectors}, cf. the bosonic case in Ref. \cite%
{Kupsch/Banerjee:2006}. The $\Lambda $-inner product of a coherent vector
with $\Psi (X,\eta )$ is, cf. Appendix \ref{ultra},%
\begin{equation}
\left( \exp \xi \mid \Psi (X,\eta )\right) =\exp \left( \xi \mid \eta +\frac{%
1}{2}X\,\xi ^{\ast }\right) ,  \label{h14}
\end{equation}%
and the $\Lambda $-inner product of an ultracoherent vector with itself is
calculated in Appendix \ref{ultra} as%
\begin{eqnarray}
&&\left( \Psi (X,\xi )\mid \Psi (X,\xi )\right) =  \notag \\
&&\sqrt{\det (I+X^{\dag }X)}\exp \left( \frac{1}{2}\left\langle \xi ^{\ast
}\parallel A\xi ^{\ast }\right\rangle +\left\langle \xi ^{\ast }\parallel
B\xi \right\rangle +\frac{1}{2}\left\langle \xi \parallel A^{\dag }\,\xi
\right\rangle \right)  \label{h13}
\end{eqnarray}%
with the operators $A=X(I+X^{\dag }X)^{-1}$ and $B=(I+XX^{\dag })^{-1}$. The
mapping $\mathcal{H}_{\Lambda }\times \mathcal{H}_{\Lambda }\ni (\xi ,\eta
)\rightarrow \left\langle \xi \parallel T\eta \right\rangle =\left( \xi
^{\ast }\mid T\eta \right) \in \Lambda _{2}$ is well defined for any bounded
operator $T\in \mathcal{L}(\mathcal{H})$ with the norm estimate $\left\Vert
\left\langle \xi \parallel T\,\eta \right\rangle \right\Vert _{\Lambda }\leq
\left\Vert T\right\Vert \left\Vert \xi \right\Vert _{\otimes }\left\Vert
\eta \right\Vert _{\otimes }$. If $T$ is a skew symmetric operator, the form
$\left\langle \xi \parallel T\eta \right\rangle $ is symmetric in $\xi $ and
$\eta $. The exponential series of these $\Lambda _{2}$-valued forms is
absolutely converging within $\Lambda $, cf. Appendix \ref{exponentials}.

\section{Weyl operators and canonical transformations\label{Weyl}}

\subsection{Weyl operators for fermions\label{Weyl_0}}

It is convenient to define Weyl operators for fermions $W(\eta )$ with$%
\,\eta \in \mathcal{H}_{\Lambda }$ first on $\mathcal{C}(\mathcal{H}%
_{\Lambda })$ by their action on exponential vectors, as it has been done
for bosons, cf. Sec. 3.1. of \cite{Kupsch/Banerjee:2006},%
\begin{equation}
W(\eta )\exp \xi :=\mathrm{e}^{-\left( \eta \mid \xi \right) -\frac{1}{2}%
\left( \eta \mid \eta \right) }\exp \,\left( \eta +\xi \right) \in \Lambda _{%
\bar{0}}\otimes \mathcal{C}(\mathcal{H}_{\Lambda }).  \label{w1}
\end{equation}%
The operators $W(\eta )$ form a group with $W(0)=id$, and%
\begin{equation}
W(\xi )W(\eta )=\mathrm{e}^{-i\omega (\xi ,\eta )}W(\xi +\eta ),\;\xi ,\eta
\in \mathcal{H}_{\Lambda }.  \label{w2}
\end{equation}%
Thereby $\omega (\xi ,\eta )$ is the $\mathbb{R}$-bilinear antisymmetric
form
\begin{equation}
\mathcal{H}_{\Lambda }\times \mathcal{H}_{\Lambda }\ni (\xi ,\eta
)\rightarrow \omega (\xi ,\eta )=\frac{1}{2i}\left( \left( \xi \mid \eta
\right) -\left( \eta \mid \xi \right) \right) \in \Lambda _{2}\subset
\Lambda _{\bar{0}}.  \label{w3}
\end{equation}%
The relation (\ref{w2}) implies that $W(\eta )$ is invertible with $%
W^{-1}(\eta )=W(-\eta )$. The expectation of the Weyl operator between
exponential vectors is calculated as\newline
$\left( \exp \xi \mid W(\eta )\,\exp \xi \right) =\exp \left( \left( \xi
\mid \xi \right) +\left( \xi \mid \eta \right) -\left( \eta \mid \xi \right)
-\frac{1}{2}\left( \eta \mid \eta \right) \right) =\left( W(-\eta )\exp \xi
\mid \exp \xi \right) $. \newline These identities imply the relation $W^{+}(\eta
)=W(-\eta )$, valid on $\mathcal{C}(\mathcal{H}_{\Lambda })$.

For fixed $\eta \in \mathcal{H}_{\Lambda }$ the operators $\mathbb{R}\ni
t\rightarrow W(t\eta )$ form a one parameter group. The generator $D_{\eta }$
of this group follows from (\ref{w1}) as $D_{\eta }\exp \xi =\frac{d}{dt}%
W(t\eta )\exp \xi \mid _{t=0}=-\left( \eta \mid \xi \right) \exp \xi +\eta
\cdot \exp \xi $, i.e. $D_{\eta }=-b^{-}(\eta )+b^{+}(\eta ),$ where $b^{\pm
}(\eta )$ are the creation and annihilation operators (\ref{h10}). Since
these operators are bounded operators on $\mathcal{A}^{\Lambda }(\mathcal{H}%
) $, the Weyl operator can be defined by the exponential series expansion%
\begin{equation}
W(\eta )=\exp \left( b^{+}(\eta )-b^{-}(\eta )\right) =\mathrm{e}^{-\frac{1}{%
2}\left( \eta \mid \eta \right) }\exp \left( b^{+}(\eta )\right) \exp \left(
-b^{-}(\eta )\right)  \label{w5}
\end{equation}%
for $\eta \in \mathcal{H}_{\Lambda }$ as bounded operators on the module
Fock space $\mathcal{A}^{\Lambda }(\mathcal{H})$. The operator (\ref{w5}) is
a regular operator in the sense of Definition \ref{regular}, and it agrees
on $\mathcal{C}(\mathcal{H}_{\Lambda })$ with (\ref{w1}). As a consequence
of Lemma \ref{regOp1} the relations $W^{-1}(\eta )=W(-\eta )=W^{+}(\eta )$
are true on $\mathcal{A}^{\Lambda }(\mathcal{H})$, and the following
identity is valid for $\eta \in \mathcal{H}_{\Lambda }$ and all tensors $\Xi
,\Upsilon \in \mathcal{A}^{\Lambda }(\mathcal{H})$%
\begin{equation}
\left( W(\eta )\,\Xi \mid W(\eta )\,\Upsilon \right) =\left( W^{+}(\eta
)W(\eta )\,\Xi \mid \Upsilon \right) =\left( \Xi \mid \Upsilon \right) .
\label{w6}
\end{equation}%
In Appendix \ref{ultra} we calculate the action of $W(\eta )$ on the
ultracoherent vectors (\ref{h12}) and obtain the formula%
\begin{equation}
W(\eta )\,\Psi (X,\xi )=\mathrm{e}^{-\frac{1}{2}\left( \eta \mid \eta
\right) +\frac{1}{2}\left( \eta \mid X\eta ^{\ast }-2\xi \right) }\Psi
(X,\xi +\eta -X\eta ^{\ast })  \label{w7}
\end{equation}%
with $X\in \mathcal{L}_{2}^{-}$ and $\xi ,\eta \in \mathcal{H}_{\Lambda }$.
Finally, we list two properties of Weyl operators that are used in the
subsequent Sections:

\begin{enumerate}
\item Let $P$ be the projection operator onto the closed subspace $\mathcal{F%
}\subset \mathcal{H}$. If the restriction $PS=SP$ of the operator $S\in
\mathcal{L}(\mathcal{H})$ is a unitary operator on $\mathcal{F}$, and if $%
\eta $ is an element of $\Lambda \widehat{\otimes }\mathcal{F}\subset
\mathcal{H}_{\Lambda }$, then the following identity is true%
\begin{equation}
\hat{\Gamma}(S)W(\eta )\hat{\Gamma}(S^{\dag }P)=W(S\eta )\hat{\Gamma}(P).
\label{w8}
\end{equation}

\item Assume $P_{j},\,j=1,2$, are projection operators onto orthogonal
subspaces $\mathcal{H}_{j}=P_{j}\mathcal{H}$ of the Hilbert space $\mathcal{H%
}$, and $P=P_{1}+P_{2}$ denotes the projection operator onto $\mathcal{H}%
_{1}\oplus \mathcal{H}_{2}$. Then the Weyl relation (\ref{w2}) implies the
factorization into $W(P\eta )=W(P_{1}\eta )W(P_{2}\eta )=W(P_{2}\eta
)W(P_{1}\eta )$. Under the additional assumption that $\Xi \in \mathcal{A}%
^{\Lambda }(\mathcal{H})$ is the product $\Xi =\Xi _{1}\circ \Xi _{2}$ of
the tensors $\Xi _{j}\in \mathcal{A}^{\Lambda }(\mathcal{H}_{j}),\,j=1,2$,
where $\Xi _{1}$ has the parity $\pi (\Xi _{1})=k\in \left\{ 0,1\right\} $
we obtain
\begin{equation}
W(P\eta )\,\Xi =\left( W(P_{1}\eta )\,\Xi _{1}\right) \circ \left(
W((-1)^{k}P_{2}\eta )\,\Xi _{2}\right) .  \label{w10}
\end{equation}
\end{enumerate}

The identity (\ref{w8}) is a consequence of the transformation behaviour $%
\Gamma (S)a^{\pm }(h)\Gamma (S^{\dag }P)=a^{\pm }(Sh)\Gamma (P),\,h\in
\mathcal{F}$, of the creation/annihilation operators in $\mathcal{A}(%
\mathcal{H})$ under unitaritary transformations of $\mathcal{H}$. The
identity (\ref{w10}) follows from (\ref{w5}) and the relations $b^{\pm
}(P_{1}\eta )\,\Xi =\newline
\left( b^{\pm }(P_{1}\eta )\,\Xi _{1}\right) \circ \Xi _{2},\,b^{\pm
}(P_{2}\eta )\,\Xi =(-1)^{k}\Xi _{1}\circ \left( b^{\pm }(P_{2}\eta )\,\Xi
_{2}\right) $ for the module creation/annihilation operators (\ref{h10}).

\subsection{Canonical transformations\label{canonical}}

The canonical anticommutation relations (CAR) for the Fock space creation
and annihilation operators $\left\{ a^{-}(f),a^{+}(g)\right\} \equiv
a^{-}(f)a^{+}(g)+a^{+}(g)a^{-}(f)=\left( f\mid g\right) I$ and $\left\{
a^{+}(f),a^{+}(f)\right\} =\left\{ a^{-}(f),a^{-}(f)\right\} =0$ can be
written in the more condensed form using the antihermitean difference $%
\Delta(f):=a^{+}(f)-a^{-}(f).$ This difference is a function on $\mathcal{H}%
_{\mathbb{R}}$ -- the underlying real space of $\mathcal{H}$ -- which has
the inner product $\mathcal{H}_{\mathbb{R}}\times\mathcal{H}_{\mathbb{R}%
}\ni(f,g)\rightarrow\left( f\mid g\right) _{\mathbb{R}}=\mathrm{Re}\,\left(
f\mid g\right) \in\mathbb{R}$. The anticommutation relations%
\begin{equation}
\Delta(f)\Delta(g)+\Delta(g)\Delta(f)=-2\left( f\mid g\right) _{\mathbb{R}%
}I\;\mathrm{with}\;f,g\in\mathcal{H}_{\mathbb{R}}  \label{can1}
\end{equation}
are equivalent to the CAR of the creation and annihilation operators. In
Sec. \ref{rep} we construct unitary operators $T$ on $\mathcal{A}(\mathcal{H}%
)$ that implement $\mathbb{R}$-linear transformations of the creation and
annihilation operators%
\begin{equation}
Ta^{\pm}(f)T^{\dagger}=a^{\pm}(Uf)-a^{\mp}(Vf^{\ast})  \label{can2}
\end{equation}
such that the canonical anticommutation relations remain unchanged. Here $U$
and $V$ are bounded linear transformations on $\mathcal{H}$. The
transformations $a^{\pm}(f)\rightarrow a^{\pm}(Uf)-a^{\mp}(Vf^{\ast})$ are
called canonical transformations or \textit{Bogoliubov transformations}.
Using the operators $\Delta(f),\,f\in\mathcal{H}_{\mathbb{R}},$ the
transformation rules (\ref{can2}) can be combined into
\begin{equation}
T\Delta(f)T^{\dagger}=\Delta(Uf+Vf^{\ast}).  \label{can2a}
\end{equation}
The mapping
\begin{equation}
f\in\mathcal{H}_{\mathbb{R}}\rightarrow R(U,V)f:=Uf+Vf^{\ast}\in \mathcal{H}%
_{\mathbb{R}},  \label{can3}
\end{equation}
with $U,\,V\in\mathcal{L}(\mathcal{H})$ is the general continuous $\mathbb{R}
$-linear mapping on $\mathcal{H}_{\mathbb{R}}$. The canonical
anticommutation relations (\ref{can1}) are preserved, if the inner product
of $\mathcal{H}_{\mathbb{R}}$ is invariant against this mapping, i.e. if
\begin{equation}
\left( R(U,V)f\mid R(U,V)g\right) _{\mathbb{R}}=\left( Uf+Vf^{\ast
}\parallel Ug+Vg^{\ast}\right) _{\mathbb{R}}=\left( f\mid g\right) _{\mathbb{%
R}}  \label{can4}
\end{equation}
holds for all $f,g\in\mathcal{H}_{\mathbb{R}}$. An invertible operator $%
R(U,V)$ which satisfies this identity is an orthogonal transformation of the
space $\mathcal{H}_{\mathbb{R}}$. Such transformations actually form a
group, which is discussed in more detail in Sec. \ref{orthog}. For infinite
dimensional Hilbert spaces $\mathcal{H}$ -- needed for quantum field theory
-- an additional constraint turns out to be necessary: In order to obtain a
unitary representation, which has a vacuum state, the operator $V$ has to be
a HS operator \cite{Friedrichs:1953, Shale/Stinespring:1965, BSZ:1992}. This
restricted orthogonal group of the space $\mathcal{H}_{\mathbb{R}}$ will be
called $\mathcal{O}_{2}(\mathcal{H}_{\mathbb{R}})$.

There exists a large number of publications, in which unitary
representations of this group on $\mathcal{A}(\mathcal{H})$ are
investigated, cf. e.g. \cite{Friedrichs:1953, Berezin:1966,
Ruijsenaars:1978, Ottesen:1995}. It is the aim of this paper to construct a
ray representation of the group $\mathcal{O}_{2}(\mathcal{H}_{\mathbb{R}})$
with the methods of superanalysis as presented in Sects. \ref{Fock} and \ref%
{Weyl}. That amounts to the construction of a representation by continuous
operators $\hat{T}(R),\,R\in\mathcal{O}_{2}(\mathcal{H}_{\mathbb{R}})$,
acting on the module Fock space $\mathcal{A}^{\Lambda}(\mathcal{H})$ with
the following property: $\hat{T}(R)$ is the product $\hat{T}%
(R)=\kappa_{0}\otimes T(R)$ where $T(R)$ is a unitary operator on $\mathcal{A%
}(\mathcal{H})$. Using the creation and annihilation operators (\ref{h10})
and $D_{\xi}=b^{+}(\xi)-b(\xi),\,\xi \in\mathcal{H}_{\Lambda},$ the
canonical anticommutation relations (\ref{can1}) are equivalent to the
commutation relations $D_{\xi}D_{\eta }-D_{\eta}D_{\xi}=-2i\omega(\xi,\eta)$
with the antisymmetric form (\ref{w3}) of the Weyl relations. Extending the
operators (\ref{can3}) to $\mathbb{R}$-linear operators on $\mathcal{H}%
_{\Lambda}$ by the rule
\begin{equation}
R(U,V)\left( \mu\otimes f\right) =\mu\otimes Uf+\mu^{\ast}\otimes Vf^{\ast
},\;\mu\in\Lambda_{1},\,f\in\mathcal{H},  \label{can5}
\end{equation}
we obtain the relation%
\begin{equation}
\omega(R\,\xi,R\,\eta)=\omega(\xi,\eta),\,\xi,\eta\in\mathcal{H}_{\Lambda },
\label{can6}
\end{equation}
which is equivalent to (\ref{can4}). Hence a unitary operator $T(R)$ on $%
\mathcal{A}(\mathcal{H})$ generates the Bogoliubov transformation (\ref{can2}%
) if%
\begin{equation}
\widehat{T}(R)W(\xi)=W(R\,\xi)\widehat{T}(R)  \label{can7}
\end{equation}
is true with $\widehat{T}(R)=\kappa_{0}\otimes T(R)$ for all $\xi \in%
\mathcal{H}_{\Lambda}$.

\section{The orthogonal group\label{orthog}}

\subsection{Definition\label{orthogdef}}

The following proposition is well known:

\begin{proposition}
The transformation (\ref{can3}) $\mathcal{H}_{\mathbb{R}}\ni f\rightarrow
R(U,V)f=Uf+Vf^{\ast}\in\mathcal{H}_{\mathbb{R}}$ is an invertible linear
transformation that preserves the inner product $\left( f\parallel g\right)
_{\mathbb{R}}=\mathrm{Re}\,\left( f\mid g\right) $ of $\mathcal{H}_{\mathbb{R%
}}$ if and only if $U$ and $V$ satisfy the identities%
\begin{align}
UU^{\dagger}+VV^{\dagger} & =I=U^{\dagger}U+V^{T}\bar{V},  \label{orth1a} \\
UV^{T}+VU^{T} & =0=U^{\dagger}V+V^{T}\bar{U}.  \label{orth1b}
\end{align}
\end{proposition}

These transformations form the group of orthogonal transformations, which
will be denoted by $\mathcal{O}(\mathcal{H}_{\mathbb{R}})$. The
multiplication rule is%
\begin{equation}
R(U_{2},V_{2})R(U_{1},V_{1})=R(U_{2}U_{1}+V_{2}\bar{V}_{1},U_{2}V_{1}+V_{2}%
\bar{U}_{1}).  \label{orth2}
\end{equation}
The identity of the group is $R(I,0)$, and the inverse of $R(U,V)$ is
\begin{equation}
R^{-1}(U,V)=R(U^{\dagger},V^{T}).  \label{orth3}
\end{equation}

In order to derive a unitary representation of the orthogonal group on the
Fock space $\mathcal{A}(\mathcal{H})$ an additional constraint is necessary
if $\dim\mathcal{H}$ is infinite: The operator $V$ has to be a HS operator
\cite{Friedrichs:1953, Shale/Stinespring:1965, BSZ:1992}. The subset of all
transformations $R(U,V)$ with a HS operator $V$ forms a subgroup, which is
denoted as $\mathcal{O}_{2}(\mathcal{H}_{\mathbb{R}})$. In the sequel we
only consider transformations with this restriction, and \textquotedblleft
orthogonal transformation\textquotedblright\ always means a transformation
in $\mathcal{O}_{2}(\mathcal{H}_{\mathbb{R}})$. The relations (\ref{orth1a})
and the HS condition for $V$ imply

\begin{enumerate}
\item $U$ and $U^{\dag }$ are Fredholm operators, i.e. the image spaces $%
\mathrm{ran}\,U$ and $\mathrm{ran}\,U^{\dag }$ closed, and the spaces $\ker
U $ and $\mathrm{\ker }\,U^{\dagger }$ have finite dimension \cite{GGK:1990}.

\item The spaces $\ker U$ and $\ker U^{\dagger }$ have the same dimension%
\begin{equation}
\dim \left( \ker U\right) =\dim \left( \ker U^{\dagger }\right) <\infty .
\label{orth4}
\end{equation}
\end{enumerate}

The second statement follows from the identities (\ref{orth1a}), which imply
$\ker U=\left\{ f\mid V^{T}\bar{V}f=f\right\} $ and $\ker U^{\dagger
}=\left\{ f\mid VV^{\dagger}f=f\right\} $. The positive trace class
operators $VV^{\dagger}$ and $V^{T}\bar{V}$ have coinciding eigenvalues
within the interval $0<\lambda\leq1$, and the corresponding eigenspaces have
the same finite dimension.

A suitable topology for $\mathcal{O}_{2}(\mathcal{H}_{\mathbb{R}})$ is given
by the norm%
\begin{equation}
\left\Vert R(U,V)\right\Vert :=\left\Vert U\right\Vert +\left\Vert
V\right\Vert _{2},  \label{orth5}
\end{equation}
where $\left\Vert U\right\Vert $ is the operator norm of $U\in\mathcal{L}(%
\mathcal{H})$ and $\left\Vert V\right\Vert _{2}$ is the HS norm of $V\in%
\mathcal{L}_{2}(\mathcal{H})$. With this topology $\mathcal{O}_{2}(\mathcal{H%
}_{\mathbb{R}})$ is a topological group, which has two connected components,
cf. Refs. \cite{Araki:1987} \S\ 6(1), \cite{Ottesen:1995} section 2.4.
Thereby the identity component of $\mathcal{O}_{2}(\mathcal{H}_{\mathbb{R}})$
is given by all transformations for which the space $\mathrm{\ker}\,U$ has
an even dimension.

The subset of all operators $R(S,0)$ with a unitary operator $S\in \mathcal{U%
}(\mathcal{H})$ forms the subgroup of all $\mathbb{C}$-linear orthogonal
transformations, which is called $\mathcal{Q}$ in the sequel.

\subsection{Orthogonal transformations $R(U,V)$ with $\ker U\neq\left\{
0\right\} $ \label{noninv}}

If the operator $U$ is not invertible, the Hilbert space can be decomposed
into the orthogonal sums $\mathcal{H}=\mathcal{H}_{0}\oplus\mathcal{H}_{1}=%
\mathcal{F}_{0}\oplus\mathcal{F}_{1}$ of the following closed subspaces $%
\mathcal{H}_{0}=\ker U^{\dagger},\,\mathcal{H}_{1}=\mathrm{ran}\,U=\left(
\ker U^{\dagger}\right) ^{\bot}$ and $\mathcal{F}_{0}=\ker U,\,\mathcal{F}%
_{1}=\mathrm{ran}\,U^{\dagger}=\left( \ker U\right) ^{\bot}$. The projection
operators onto $\mathcal{H}_{k}$ are called $P_{k}$, and the projection
operators onto $\mathcal{F}_{k}$ are called $Q_{k},\,k=0,1$. The identities (%
\ref{orth1a}) imply the relations $P_{0}VV^{\dagger}=VV^{\dagger}P_{0}=P_{0}$
and $\bar{Q}_{0}V^{\dagger}V=V^{\dagger}V\bar{Q}_{0}=\bar{Q}_{0}$ and the
operator norm inequalities $\left\Vert V\bar{Q}_{1}\right\Vert <1$ and $%
\left\Vert V^{\dagger}P_{1}\right\Vert <1$. From the identities (\ref{orth1b}%
) we obtain the inclusions
\begin{equation}
V\mathcal{F}_{1}^{\ast}\subset\mathcal{H}_{1},\;V^{T}\mathcal{H}_{1}^{\ast
}\subset\mathcal{F}_{1}.  \label{orth6}
\end{equation}
Therefore $P_{0}V=V\bar{Q}_{0}$ is an isometric mapping from $\mathcal{F}%
_{0}^{\ast}$ onto $\mathcal{H}_{0}$, and $\bar{Q}_{0}V^{\dagger}=V^{\dagger
}P_{0}$ is the inverse isometry from $\mathcal{H}_{0}$ onto $\mathcal{F}%
_{0}^{\ast}$.

The orthogonal transformation (\ref{can3}) $R(U,V)$ is the sum of two $%
\mathbb{R}$-linear mappings%
\begin{equation}
R(U,V)=R(0,P_{0}V)+R(U,P_{1}V).  \label{orth8}
\end{equation}
The transformation $R(0,P_{0}V)=R(0,V\bar{Q}_{0})$ is an isometric mapping
from $\mathcal{F}_{0\mathbb{R}}$ onto $\mathcal{H}_{0\mathbb{R}}$ with
inverse $R(0,Q_{0}V^{T})$, and $R(U,P_{1}V)=R(P_{1}U,V\bar{Q}_{1})$ is an
isometric mapping from $\mathcal{F}_{1\mathbb{R}}$ onto $\mathcal{H}_{1%
\mathbb{R}}$ with inverse $R(U^{\dagger},Q_{1}V^{T})$, cf. the identity (\ref%
{orth3}).

\subsection{The left coset space $\mathcal{O}_{2}/\mathcal{Q}$ \label{quot}}

The left coset space $\mathcal{O}_{2}/\mathcal{Q}$ is the space of all
orbits \newline
$\left\{ R(U,V)R(S,0)=R(US,V\bar{S})\mid S\in\mathcal{U}(\mathcal{H}%
)\right\} $ under right action of the subgroup $\mathcal{Q}$. This space can
serve to identify states in the Fock space representation of the orthogonal
group. If $U$ is invertible the operator $U^{\dagger-1}V^{T}=U^{%
\dagger-1}SS^{\dagger}V^{T}=(US)^{\dagger-1}(V\bar{S})^{T}$ is invariant
against the right action of this subgroup. From the identity (\ref{orth1b})
we obtain $U^{\dagger-1}V^{T}+V\bar{U}^{-1}=U^{\dagger -1}\left( V^{T}\bar{U}%
+U^{\dagger}V\right) \bar{U}^{-1}=0$. Hence%
\begin{equation}
X:=V\bar{U}^{-1}=-U^{\dagger-1}V^{T}=-X^{T}  \label{orth10}
\end{equation}
is a skew symmetric HS operator, which characterizes the orbit through $%
R(U,V)$.

\begin{remark}
\label{lift}Knowing the operator (\ref{orth10}) one can easily obtain a
transformation on the orbit. The operators $L=\left( I+XX^{\dagger}\right)
^{-\frac{1}{2}}\geq0$ and $W=X\left( I+X^{\dagger}X\right) ^{-\frac{1}{2}%
}=-W^{T}$ satisfy the conditions (\ref{orth1a}) and (\ref{orth1b}) so that $%
R(L,W)$ is an orthogonal transformation. The formula (\ref{orth10}) $W\bar {L%
}^{-1}=X$ reproduces the input $X$. Hence $R(L,W)$ is a point on the orbit
specified by $X$.
\end{remark}

If $\ker\,U\neq0$ the space $\mathcal{H}_{0}=\ker\,U^{\dagger}=\ker
\,(US)^{\dagger},\,S$ unitary, is invariant against the right action of the
group $\mathcal{Q}$. This space can therefore be used as a coordinate for an
orbit. But the operator (\ref{orth10}) is no longer defined. It is
convenient to introduce the notion of a \textit{generalized inverse operator}%
. Let $A\in\mathcal{L}(\mathcal{H})$ be an operator with closed range. Then
the generalized inverse $A^{(-1)}\in\mathcal{L}(\mathcal{H})$ of $A$ is
defined as, cf. e. g. \cite{Groetsch:1977, Ben-Israel:2003},
\begin{equation}
A^{(-1)}f:=A^{-1}Pf=\left\{
\begin{array}{l}
A^{-1}f\quad\mathrm{if}\;f\in\mathrm{ran}\,A, \\
0\quad\mathrm{if}\;f\in\ker A^{\dagger}=\left( \mathrm{ran}A\right) ^{\perp
},%
\end{array}
\right.  \label{orth12}
\end{equation}
where $P$ is the projector onto $\mathrm{ran}\,A$. The operator $A^{(-1)}$
satisfies the identities $AA^{(-1)}=P$ and $A^{(-1)}A=Q$, thereby $Q$ is the
projector onto $\mathrm{ran}A^{\dag}$. The relations $\left( A^{\dagger
}\right) ^{(-1)}=\left( A^{(-1)}\right) ^{\dagger},\;\left( \bar {A}\right)
^{(-1)}=\overline{\left( A^{(-1)}\right) }$ and $\left( A^{T}\right)
^{(-1)}=\left( A^{(-1)}\right) ^{T}$ are valid. In the sequel we write $%
A^{\dagger(-1)}$ instead of $\left( A^{\dagger}\right) ^{(-1)}$. If $S$ is a
unitary operator, then the operators $SA$ and $AS$ have also a closed range,
and the following identities are true%
\begin{equation}
\left( SA\right) ^{(-1)}=A^{(-1)}S^{\dag},\;\left( AS\right)
^{(-1)}=S^{\dag}A^{(-1)}.  \label{orth13}
\end{equation}
Using the notation of the definition (\ref{orth12}) the generalization of (%
\ref{orth10}) is the HS operator%
\begin{equation}
X:=V\bar{U}^{(-1)}=V\bar{U}^{-1}\bar{P}_{1}\overset{(\ref{orth6})}{=}P_{1}V%
\bar{U}^{-1}\bar{P}_{1}=-V^{T}U^{\dagger(-1)}.  \label{orth11}
\end{equation}
The skew symmetry is again a consequence of (\ref{orth1b}): $%
X+X^{T}=P_{1}\left( V\bar{U}^{-1}+U^{\dagger-1}V^{T}\right) \bar{P}%
_{1}=U^{\dagger(-1)}\left( U^{\dagger}V+V^{T}\bar{U}\right) \bar{U}^{(-1)}%
\overset{(\ref{orth1b})}{=}0$. The operator (\ref{orth11}) is invariant
against right action of the subgroup $\mathcal{Q}:V\bar{U}^{(-1)}=V\bar {S}%
S^{T}\bar{U}^{(-1)}\overset{(\ref{orth13})}{=}(V\bar{S})\overline{\left(
US\right) }^{(-1)}$, and it can serve as second coordinate for an orbit.

\section{Representations\label{rep}}

This section presents the central result of this paper: the construction of
a representation of $\mathcal{O}_{2}(\mathcal{H}_{\mathbb{R}})$. We are
looking for operators $T(R)$, which have the properties%
\begin{equation}
R\in\mathcal{O}_{2}(\mathcal{H}_{\mathbb{R}})\rightarrow T\left( R\right) \;%
\mathrm{unitary\,operator\,on}\;\mathcal{A}(\mathcal{H}),\;T(id)=I_{\mathcal{%
A}(\mathcal{H})},  \label{r1}
\end{equation}
and which implement canonical transformations, cf. (\ref{can2a}),%
\begin{equation}
T(R)\Delta(f)T^{\dag}(R)=\Delta(Rf),\;f\in\mathcal{H}_{\mathbb{R}}\text{.}
\label{r2}
\end{equation}
Since the CAR algebra generated by $\left\{ \Delta(f)\mid f\in\mathcal{H}_{%
\mathbb{R}}\right\} $ is irreducible within $\mathcal{A}(\mathcal{H})$, the
unitary operators $T(R)$ are determined by (\ref{r1}) up to a phase factor.
Moreover, the group laws of $\mathcal{O}_{2}(\mathcal{H}_{\mathbb{R}})$ and
the relation (\ref{r1}) imply $T(R_{2})T(R_{1})\Delta(f)T^{\dag}(R_{1})T^{%
\dag}(R_{2})=T(R_{2})\Delta(R_{1}f)T^{\dag}(R_{2})=\Delta(R_{2}R_{1}f)$, and
the product rule
\begin{equation}
T(R_{2})T(R_{1})=\chi(R_{2},R_{1})T(R_{2}R_{1})\;\mathrm{with}\;\chi
(R_{2},R_{1})\in\mathbb{C},\,\left\vert \chi(R_{2},R_{1})\right\vert =1
\label{r3}
\end{equation}
follows. Hence $R\in\mathcal{O}_{2}(\mathcal{H}_{\mathbb{R}})\rightarrow
T\left( R\right) $ is a unitary ray representation. The multiplier $\chi$
satisfies the cocycle identity $\chi(R_{3},R_{2})\chi(R_{3}R_{2},R_{1})=%
\chi(R_{3},R_{2}R_{1})\chi(R_{2},R_{1})$ as consequence the associativity of
the operator product.

A simple class of canonical transformations arise if $R(U,V)$ is $\mathbb{C}$%
-linear, i.e. for $U=S\in\mathcal{U}(\mathcal{H})$ and $V=0$. Then the
condition (\ref{can4}) is fulfilled and $T(S,0)=\Gamma(S)$ is a unitary
operator on $\mathcal{A}(\mathcal{H})$ that leads to the transformation $%
Ta^{\pm}(f)T^{+}=a^{\pm}(Sf)$ of the creation/annihilation operators in
agreement with the rule (\ref{can2}). The group laws and (\ref{r3}) imply
the relations%
\begin{equation}
\Gamma(S)T(U,V)=\chi\,T(SU,SV),\;T(U,V)\Gamma(S)=\chi\,T(US,V\bar {S})
\label{r4}
\end{equation}
with phase factors $\chi$.

In Sec. \ref{invertible} the case of orthogonal transformations $R(U,V)$
with an invertible operator $U$ is investigated. For these transformations
superanalytic methods are adequate. The construction is first given for
operators $\hat{T}\left( R\right) =\hat{T}\left[ U,V\right] $ on the module
space $\mathcal{A}^{\Lambda}(\mathcal{H})$. These operators have the
structure $\hat{T}\left( R\right) =\kappa_{0}\otimes T\left( R\right) $ with
$T\left( R\right) \in\mathcal{L}\left( \mathcal{A}(\mathcal{H})\right) $,
and the pull-back to operators $T\left( R\right) $ on the Fock space $%
\mathcal{A}(\mathcal{H})$ can be easily performed. The case of orthogonal
transformations with operators $U$, which are not invertible, is treated in
Sec. \ref{nonInv}. Then the construction given in Sec. \ref{invertible} is
only possible on the subspace $\mathcal{A}^{\Lambda }(\mathcal{F}_{1}),\,%
\mathcal{F}_{1}=\mathrm{ran}\,U^{\dagger}$, and a separate investigation
concerning the finite dimensional space $\mathcal{A}(\mathcal{F}_{0}),\,%
\mathcal{F}_{0}=\ker U$, is necessary.

\subsection{Transformations $R(U,V)$ with invertible $U$ \label{invertible}}

\subsubsection{The ansatz\label{ansatz}}

Let $R(U,V)$ be an orthogonal transformation of the group $\mathcal{O}_{2}(%
\mathcal{H}_{\mathbb{R}})$ with an invertible operator $U$. In
correspondence to the case of bosons -- cf. Sec. 5.1. in \cite%
{Kupsch/Banerjee:2006} -- the representation $\hat{T}(R)=\hat{T}\left[ U,V%
\right] $ of this transformation is now defined on the set of exponential
vectors $\exp\xi,\,\xi\in\mathcal{H}_{\Lambda}$, by the ansatz%
\begin{equation}
\hat{T}\left[ U,V\right] \exp\xi:=c_{X}\,\exp\left( -\frac{1}{2}\left\langle
\xi\parallel V^{\dagger}U^{\dagger-1}\xi\right\rangle \right)
\Psi(X,U^{\dagger-1}\xi)\in\mathcal{C}^{\Lambda}(\mathcal{H}_{\Lambda })
\label{r11}
\end{equation}
with the skew symmetric HS operator (\ref{orth10}) $X=V\bar{U}%
^{-1}=-U^{\dagger-1}V^{T}$ and the normalization constant
\begin{equation}
c_{X}=\left( \det\left( I+X^{\dagger}X\right) \right)
^{-1/4}=c_{X^{\dagger}}.  \label{r18}
\end{equation}
The tensor $\Psi$ is an ultracoherent vector (\ref{h12}). The operator $%
Y:=-V^{\dagger}U^{\dagger-1}=\bar{U}^{-1}\bar{V}$ is a skew symmetric HS
operator, and the quadratic form $\left\langle \xi\parallel Y\,\xi
\right\rangle $ is an element of $\Lambda_{2}$. The exponential series $%
\exp\left( \frac{1}{2}\left\langle \xi\parallel Y\,\xi\right\rangle \right) $
converges within $\Lambda_{\bar{0}}$, cf. Appendix \ref{exponentials}. The
right side of (\ref{r11}) is therefore a well defined element of $\mathcal{C}%
^{\Lambda}(\mathcal{H}_{\Lambda})$. The exponential vector $\exp\xi$ and (%
\ref{r11}) are entire analytic functions in the variable $\xi\in\mathcal{H}%
_{\Lambda}$. Using (\ref{r11}) and (\ref{h14}) we derive the identity%
\begin{equation}
\begin{array}{l}
\left( \exp\xi\mid\hat{T}(R)\exp\eta\right) =c_{X}\,\left( \exp\xi\mid \Psi(V%
\bar{U}^{-1},U^{\dag-1}\eta\right) \exp\left( -\frac{1}{2}\left\langle
\eta\parallel V^{\dag}U^{\dag-1}\eta\right\rangle \right) = \\
c_{X}\exp\left( -\frac{1}{2}\left\langle \eta\parallel
V^{\dag}U^{\dag-1}\eta\right\rangle \right) \exp\left( U^{-1}\xi\mid\eta-%
\frac{1}{2}V^{T}\xi^{\ast}\right) =\left( \hat{T}(R^{-1})\exp\xi\mid\exp\eta%
\right)%
\end{array}
\label{r13}
\end{equation}
with $R^{-1}(U,V)=R(U^{\dagger},V^{T})$, cf. (\ref{orth3}). Since $\xi$ and $%
\eta$ are arbitrary vectors in $\mathcal{H}_{\Lambda}$, the identity (\ref%
{r13}) implies that (\ref{r11}) has a $\Lambda$-linear extension onto $%
\mathcal{C}^{\Lambda}(\mathcal{H}_{\Lambda})$ that satisfies $\hat {T}%
(R)\left( \lambda\exp\xi\right) =\lambda\left( \hat{T}(R)\exp\xi\right)
,\;\lambda\in\Lambda$. Moreover, we obtain the operator identity%
\begin{equation}
\hat{T}^{+}(R)\equiv\left( \hat{T}(R)\right) ^{+}=\hat{T}(R^{-1})
\label{r14}
\end{equation}
on the space $\mathcal{C}^{\Lambda}(\mathcal{H}_{\Lambda})$. Using the
formula (\ref{h13}) the $\Lambda$-inner product $\left( \hat{T}(R)\exp\xi\mid%
\hat {T}(R)\exp\xi\right) $ is calculated with the result $\left( \hat{T}%
(R)\exp\xi\mid\hat{T}(R)\exp\xi\right) =\exp\left( \xi\mid\xi\right) $ for
all $R\in\mathcal{O}_{2}(\mathcal{H}_{\mathbb{R}})$ and $\xi\in\mathcal{H}%
_{\Lambda}$. By analyticity arguments this identity implies $\left( \hat {T}%
(R)\exp\xi\mid\hat{T}(R)\exp\eta\right) =\left( \exp\xi\mid\hat{T}^{+}(R)%
\hat{T}(R)\exp\eta\right) =\left( \exp\xi\mid\exp\eta\right) $ for all $%
\xi,\eta\in\mathcal{H}_{\Lambda}$, and we obtain the operator identity
\begin{equation}
\hat{T}^{+}(R)\hat{T}(R)=id  \label{r16}
\end{equation}
on $\mathcal{C}^{\Lambda}(\mathcal{H}_{\Lambda})$. The relations (\ref{r14})
and (\ref{r16}), which are valid for all $R\in\mathcal{O}_{2}(\mathcal{H}_{%
\mathbb{R}})$, imply that $\hat{T}(R)$ is an invertible operator on $%
\mathcal{C}^{\Lambda}(\mathcal{H})$ with inverse $\left( \hat{T}(R)\right)
^{-1}=\hat{T}^{+}(R)$.

To obtain the pull-back of the operators $\hat{T}(R)$ to the Fock space $%
\mathcal{A}(\mathcal{H})$ we observe that the ansatz (\ref{r11}) has the
following structure: $\hat{T}(R)$ does not operate on the $\Lambda$ factors
of $\exp\xi$, hence $\hat{T}(R)$ is the product
\begin{equation}
\hat{T}(R)=\kappa_{0}\otimes T(R),  \label{r15}
\end{equation}
where $T(R)$ is an operator on $\mathcal{A}(\mathcal{H})$. The explicit form
of $T(R)$ can be recovered from (\ref{r11}) by contraction of the $\Lambda $%
-factors. For any $\lambda\in\Lambda$ we have $T(R)\,\left( \lambda\mid
\exp\xi\right) _{\Lambda}=\left( \lambda\mid\hat{T}(R)\exp\xi\right)
_{\Lambda}$. Since the $\mathbb{C}$-linear span of the set $\left\{ \left(
\lambda\mid\exp\xi\right) _{\Lambda}\right\} $ with $\lambda\in\Lambda$ and $%
\xi\in\mathcal{H}_{\Lambda}$ is dense in $\mathcal{A}(\mathcal{H})$, the
operator $T(R)$ is determined by (\ref{r11}). The factorization (\ref{r15})
and the identities (\ref{r14}) and (\ref{r16}) imply that the operator $T(R)$
has a unique extension to a unitary operator on $\mathcal{A}(\mathcal{H})$.

\begin{remark}
Starting from the expansion $\xi =\sum_{m\in \mathbb{N}}\kappa _{m}\otimes
f_{m}\in \mathcal{H}_{\Lambda }$ with an ON basis $\left\{ \kappa
_{m}\right\} $ of $\Lambda _{1}$ and vectors $f_{m}\in \mathcal{H}$, one
obtains a more explicit version of the ansatz (\ref{r11}), which allows to
derive $T(R)\left( f_{1}\wedge ...\wedge f_{n}\right) $ for arbitrary
vectors $f_{j}\in \mathcal{H},\,j=1,...,n,\,n\in \mathbb{N}$. The
calculation is given in the Appendix \ref{pull}.
\end{remark}

The formula (\ref{r11}) applies to transformations $R(U,0)$ with a unitary
operator $U=S\in\mathcal{U}(\mathcal{H})$. For these transformations it has
the simple form $\hat{T}\left[ S,0\right] \exp\xi=\exp(S\xi)$, and $\hat {T}%
\left[ S,0\right] $ is identified with $\hat{T}\left[ S,0\right] =\hat{\Gamma%
}(S):=\kappa_{0}\otimes\Gamma(S).$The unitary transformation $\Gamma(S)$ is
a canonical transformation on $\mathcal{A}(\mathcal{H})$ as already stated
in the introductory part of Sec. \ref{rep}. The operators (\ref{r11}) satisfy
the relations
\begin{equation}
\hat{\Gamma}(S)\hat{T}\left[ U,V\right] =\hat{T}\left[ SU,SV\right] \;%
\mathrm{and}\;\hat{T}\left[ U,V\right] \hat{\Gamma}(S)=\hat{T}\left[ US,V%
\bar{S}\right] ,  \label{r5}
\end{equation}
which imply the corresponding identities for $T\left[ U,V\right] $.

If $U$ is not invertible, the ansatz (\ref{r11}) is not defined for vectors $%
\xi \in \mathcal{H}_{\Lambda }=\Lambda _{1}\widehat{\otimes }\mathcal{H}$
with a component in $\Lambda _{1}\widehat{\otimes }\mathcal{F}_{0},\mathcal{F%
}_{0}=\ker U$. But we can obtain a mapping $\hat{T}_{1}=\hat{T}_{1}\left[
U,P_{1}V\right] =\hat{T}_{1}(R_{1})$ from $\mathcal{C}^{\Lambda }(\Lambda
_{1}\widehat{\otimes }\mathcal{F}_{1})=\hat{\Gamma}(Q_{1})\,\mathcal{C}%
^{\Lambda }(\mathcal{H}_{\Lambda })$ onto $\mathcal{C}^{\Lambda }(\Lambda
_{1}\widehat{\otimes }\mathcal{H}_{1})=\hat{\Gamma}(P_{1})\,\mathcal{C}%
^{\Lambda }(\mathcal{H}_{\Lambda })$ that represents the transformation $%
R_{1}=R(U,P_{1}V)$ of Eq. (\ref{orth8}).\footnote{%
For the definition of the spaces and the operators see Sec. \ref{noninv}.}
On the restricted set of coherent vectors $\exp \xi $ with $\xi \in \Lambda
_{1}\widehat{\otimes }\mathcal{F}_{1},\,\mathcal{F}_{1}=\mathrm{ran}%
\,U^{\dagger }$, we define the operator
\begin{equation}
\hat{T}_{1}\left[ U,P_{1}V\right] \exp \xi :=c_{X}\,\Psi (X,U^{\dagger
(-1)}\xi )\exp \left( -\frac{1}{2}\left\langle \xi \parallel V^{\dagger
}U^{\dagger (-1)}\xi \right\rangle \right) \in \mathcal{C}^{\Lambda
}(\Lambda _{1}\widehat{\otimes }\mathcal{H}_{1})  \label{r19}
\end{equation}%
with the skew symmetric operator (\ref{orth11}) $X$ and the normalization
constant (\ref{r18}). The right side of (\ref{r19}) depends only on the
operators $U=P_{1}U$ and $P_{1}V.$ Arguments as used for (\ref{r11}) yield
that the ansatz (\ref{r19}) can be extended to an invertible mapping from $%
\mathcal{A}^{\Lambda }(\mathcal{F}_{1})$ onto $\mathcal{A}^{\Lambda }(%
\mathcal{H}_{1})$.

There is another approach to the operator $\hat{T}_{1}$. Let $U_{0}$ be a
partial isometry from $\mathcal{F}_{0}$ onto $\mathcal{H}_{0}$, i.e. the
operator $U_{0}$ satisfies $U_{0}U_{0}^{\dagger }=P_{0}$ and $U_{0}^{\dagger
}U_{0}=Q_{0}$. Then the transformation $R(U+U_{0},P_{1}V)$ is an orthogonal
transformation with invertible first argument $(U+$ $%
U_{0})^{-1}=U^{(-1)}+U_{0}^{\dagger }$, and the operator $\hat{T}\left[
U+U_{0},P_{1}V\right] $ can be defined by the ansatz (\ref{r11}). The
mapping (\ref{r19}) is the restriction of the operator $\hat{T}\left[
U+U_{0},P_{1}V\right] $ to the space $\mathcal{A}^{\Lambda }(\mathcal{F}%
_{1})=\hat{\Gamma}(Q_{1})\mathcal{A}^{\Lambda }(\mathcal{H})$%
\begin{equation}
\hat{T}_{1}\left[ U,P_{1}V\right] \,\hat{\Gamma}(Q_{1})=\hat{T}\left[
U+U_{0},P_{1}V\right] \,\hat{\Gamma}(Q_{1}).  \label{r25}
\end{equation}%
If $S$ is a unitary operator on $\mathcal{H}$, we infer from (\ref{r19}), or
from (\ref{r5}) and (\ref{r25}), the following identity, which is needed in
Sec. \ref{general},%
\begin{equation}
\hat{T}_{1}\left[ U,P_{1}V\right] \,\hat{\Gamma}(Q_{1})=\hat{T}_{1}\left[
US,P_{1}V\bar{S}\right] \hat{\Gamma}(S^{\dagger }Q_{1})\mathrm{.}
\label{r20}
\end{equation}%
As a consequence of (\ref{r15}) and (\ref{r25}) the operator $\hat{T}_{1}$
has the structure $\hat{T}_{1}(R_{1})=\kappa _{0}\otimes T_{1}(R_{1})$,
where $T_{1}$ is an isometric surjective mapping from $\mathcal{A}(\mathcal{F%
}_{1})$ onto $\mathcal{A}(\mathcal{H}_{1})$.

\subsubsection{The intertwining relation with Weyl operators\label{WeylBogol}%
}

The ansatz (\ref{r11}) and the relation (\ref{w1}) for the Weyl operators
yield the identity%
\begin{eqnarray}
&&\hat{T}(R)W(\xi )\exp \eta =c_{X}\,\Psi \left( X,\,U^{\dagger -1}(\xi
+\eta )\right) \times  \notag \\
&&\exp \left( -\left( \xi \mid \eta \right) -\frac{1}{2}\left( \xi \mid \xi
\right) -\frac{1}{2}\left\langle \xi +\eta \parallel V^{\dagger }U^{\dagger
-1}(\xi +\eta )\right\rangle \right)  \label{r24}
\end{eqnarray}%
with the operator (\ref{orth10}) $X$, the normalization constant (\ref{r18}%
), and supervectors $\xi ,\eta \in \mathcal{H}_{\Lambda }$. On the other
hand $W(R\xi )\hat{T}(R)\exp \eta $ is calculated using (\ref{w7}), (\ref%
{orth1a}), (\ref{orth1b}) and (\ref{r11}) with the same result. The
intertwining relation
\begin{equation}
\hat{T}(R)W(\xi )=W(R\xi )\hat{T}(R)  \label{r21}
\end{equation}%
is therefore true on the space $\mathcal{C}(\mathcal{H}_{\Lambda })$ for all
$R(U,V)\in \mathcal{O}_{2}(\mathcal{H}_{\mathbb{R}})$ with invertible $U$
and for all supervectors $\xi \in \mathcal{H}_{\Lambda }$. Then Lemma \ref%
{regOp1} implies that this identity is true on the module Fock space $%
\mathcal{A}^{\Lambda }(\mathcal{H})$. Hence the operators (\ref{r11})
implement Bogoliubov transformations as defined in Sec. \ref{canonical}.

If $U$ is not invertible, we can derive a partial intertwining relation for
the mapping (\ref{r19}). The orthogonal transformation $R$ is restricted to
the isometric mapping $R_{1}:=R(U,P_{1}V)$ from $\mathcal{F}_{1\mathbb{R}}$
onto $\mathcal{H}_{1\mathbb{R}}$. If $\eta $ is a supervector in $\Lambda
_{1}\widehat{\otimes }\mathcal{F}_{1}$, the tensor $R_{1}\eta =U\eta
+P_{1}V\eta ^{\ast }$ is a supervector in $\Lambda _{1}\widehat{\otimes }%
\mathcal{H}_{1}$. Then the calculations given above go through with
supervectors $\xi ,\eta \in \Lambda _{1}\widehat{\otimes }\mathcal{F}_{1}$
and the skew symmetric operator (\ref{orth11}) $X$. The resulting identity $%
\hat{T}(R_{1})W(\xi )\exp \eta =W(R_{1}\xi )\hat{T}(R_{1})\exp \eta $
implies the relation%
\begin{equation}
\hat{T}_{1}(R_{1})W(\xi )\hat{\Gamma}(Q_{1})=W(R_{1}\xi )\hat{T}_{1}(R_{1})%
\hat{\Gamma}(Q_{1}).  \label{r22}
\end{equation}

\subsection{Transformations $R(U,V)$ with $\ker U\neq\left\{ 0\right\} $%
\label{nonInv}}

If $\ker U=\mathcal{F}_{0}\neq\left\{ 0\right\} $ it is convenient to split
the orthogonal transformation $R(U,V)$ into the sum (\ref{orth8}). At the
end of Sec. \ref{ansatz} we have introduced the isometric mapping $T_{1}$
from $\mathcal{A}(\mathcal{F}_{1})$ onto $\mathcal{A}(\mathcal{H}_{1})$ as
representation of the partial isometry $R_{1}=R(U,P_{1}V)$, the second term
in the sum (\ref{orth8}). In Sec. \ref{dual} we construct an isometric
mapping $T_{0}$ from $\mathcal{A}(\mathcal{F}_{0})$ onto $\mathcal{A}(%
\mathcal{H}_{0})$ that is a representation of the partial isometry $%
R_{0}=R(0,P_{0}V)$, the first term in the sum (\ref{orth8}). In Sec. \ref%
{general} we derive that the combined action of these operators is a
representation of the orthogonal transformation (\ref{orth8}) on the Fock
space $\mathcal{A}(\mathcal{H})$.

\subsubsection{A duality mapping between finite dimensional spaces\label%
{dual}}

We start with a linear partial isometry $J$ from $\mathcal{F}_{0}^{\ast }$
onto $\mathcal{H}_{0}$, i.e. $J$ satisfies the identities $J^{\dagger }J=%
\bar{Q}_{0}$ and $JJ^{\dagger }=P_{0}$. We choose an ON basis $\left\{
e_{m},\,m\in \mathbf{M}\right\} ,\,\mathbf{M}=\left\{ 1,...,n\right\} $ of $%
\mathcal{H}_{0}$. Then the vectors $\left\{ f_{m}=-J^{T}\,e_{m}^{\ast }\mid
m\in \mathbf{M}\right\} $ form an ON basis of $\mathcal{F}_{0}$. We define
the linear mapping $T_{0}\left[ J\right] $ by its action on the basis $%
\left\{ f_{\mathbf{K}},\,\mathbf{K}\subset \mathbf{M}\right\} $ of $\mathcal{%
A}(\mathcal{F}_{0})$%
\begin{equation}
T_{0}\left[ J\right] \,f_{\mathbf{K}}:=(-1)^{\tau (\mathbf{K},\mathbf{N})}e_{%
\mathbf{N\backslash K}},\;\mathbf{K}\subset \mathbf{N}.  \label{r31}
\end{equation}%
The number $\tau (\mathbf{K},\mathbf{N})$ has been defined at the end of
Sec. \ref{hilbert}. The linear extension of (\ref{r31}) is obviously an
isometric linear isomorphism $T_{0}\left[ J\right] $ between $\mathcal{A}(%
\mathcal{F}_{0})$ and $\mathcal{A}(\mathcal{H}_{0})$. The operator $T_{0}$
has the following intertwining relations with the operator $\Delta
(h)=a^{+}(h)-a^{-}(h)$%
\begin{equation}
T_{0}\left[ J\right] \Delta (h)=\Delta (R_{0}h)T_{0}\left[ J\right] \quad
\mathrm{if}\;h\in \mathcal{F}_{0}.  \label{r33}
\end{equation}%
Thereby $R_{0}:=R(0,J)$ is the $\mathbb{R}$-linear operator (\ref{can3}) $%
R(0,J)h=Jh^{\ast }$.The proof of these relations is given in Appendix \ref%
{reflections}. The irreducibility of the CAR algebras of the spaces $%
\mathcal{A}(\mathcal{F}_{0})$ and $\mathcal{A}(\mathcal{H}_{0})$ implies
that the operator $T_{0}$ is unique except for a phase factor. Hence the
constructions of the operator (\ref{r31}) based on different choices of the
basis $\left\{ e_{j}\right\} $ can disagree only in a multiplicative phase
factor.

\begin{remark}
In Ref. \cite{Ruijsenaars:1978} the mapping (\ref{r31}) has been constructed
using creation and annihilation operators. If $\mathcal{F}_{0}=\mathcal{H}%
_{0}=\mathcal{H}_{0}^{\ast}$ -- e. g. in the case of a real self-adjoint
operator $U$ -- the mapping (\ref{r31}) is a modified Hodge star operator
within $\mathcal{A}(\mathcal{H}_{0})$.
\end{remark}

The ambiguity of $T_{0}\left[ J\right] $ is an overall phase factor, which
can be fixed at the vacuum $f_{\emptyset}$ of the space $\mathcal{F}_{0}$.
With the labels $T_{0}\left[ J,\hat{e}\right] $, where $\hat{e}=T_{0}\left[
J,\hat{e}\right] \,f_{\emptyset}\in\mathcal{H}_{0}^{\wedge n}$ is the image
of the vacuum vector, the mapping $T_{0}\left[ J,\hat{e}\right] $ is
uniquely determined. With this phase convention the relation%
\begin{equation}
T_{0}\left[ J,\hat{e}\right] \Gamma(S)=T_{0}\left[ J\bar{S},\hat{e}\right]
\label{r23}
\end{equation}
is valid for all unitary operators $S\in\mathcal{U}(\mathcal{H})$.

The mapping $T_{0}\left[ J\right] $ has the natural extension $\hat{T}_{0}%
\left[ J\right] =\kappa_{0}\otimes T_{0}\left[ J\right] $ to an invertible
mapping from $\Lambda\widehat{\otimes}\mathcal{A}(\mathcal{F}_{0})$ onto $%
\Lambda\widehat{\otimes}\mathcal{A}(\mathcal{H}_{0})$. The intertwining
relation of this operator $\hat{T}_{0}(R_{0})=\hat{T}_{0}\left[ J\right] $
with the Weyl operator follows from (\ref{r33}) as%
\begin{equation}
\hat{T}_{0}(R_{0})W(\xi)\hat{\Gamma}(Q_{0})=W(R_{0}\xi)\hat{T}_{0}(R_{0})%
\hat{\Gamma}(Q_{0}).  \label{r34}
\end{equation}
Thereby $\hat{\Gamma}(Q_{0})$ is the projection operator onto the space $%
\Lambda\widehat{\otimes}\mathcal{A}(\mathcal{F}_{0})\subset\mathcal{A}%
^{\Lambda}(\mathcal{H})$.

There is a structural difference between operators $\hat{T}_{0}\left[ J%
\right] $ with an even dimension of $\mathcal{H}_{0}=\mathrm{ran}\,J\,(=\ker
U^{\dag})$ and operators with an odd dimension of this space. If $\dim%
\mathcal{H}_{0}$ is even, the operator $\hat{T}_{0}$ maps tensors of parity $%
\pi$ onto tensors with the same parity. If $\dim\mathcal{H}_{0}$ is odd, the
operator $\hat{T}_{0}$ changes the parity $\pi$ into $\pi+1\,mod\,2$. This
structural difference of $\hat{T}_{0}\left[ J\right] $ reflects the group
theoretical difference between the orthogonal transformations $R(U,V)$ with $%
\dim\left( \ker U\right) $ even or odd. If $\dim\left( \ker U\right) $ is
even (odd), the orthogonal transformation $R(U,V)$ belongs (does not belong)
to the connectivity component of the group identity, cf. \S\ 6 of \cite%
{Araki:1970}.

\subsubsection{The general case\label{general}}

The canonical transformation $\hat{T}(R)$ is defined on $\mathcal{A}%
^{\Lambda }(\mathcal{H})$ combining the operators $\hat{T}_{0}$ and (\ref%
{r19}) $\hat{T}_{1}$. Thereby the operator $J$ is identified with $P_{0}V=V%
\bar{Q}_{0}$. We have to take into account that $\hat{T}_{0}$ interchanges
the parity of states, if $n=\dim \mathcal{F}_{0}=\dim \mathcal{H}_{0}$ is an
odd number. The mapping $\hat{T}(R)=\hat{T}\left[ U,V\right] $ is now
defined on the product $\Xi =\Xi _{0}\circ \Xi _{1}\in \mathcal{A}^{\Lambda
}(\mathcal{H})$ of tensors $\Xi _{k}\in \mathcal{A}^{\Lambda }(\mathcal{F}%
_{k})=\hat{\Gamma}(Q_{k})\mathcal{A}^{\Lambda }(\mathcal{H}),\,k=0,1$, as%
\begin{equation}
\hat{T}(R)\left( \Xi _{0}\circ \Xi _{1}\right) =\left( \hat{T}_{0}(R_{0})\Xi
_{0}\right) \circ \left( \hat{T}_{1}(R_{1})\hat{\Gamma}((-1)^{n}Q_{1})\Xi
_{1}\right)  \label{r35}
\end{equation}%
with $n=\dim \mathcal{F}_{0}=\dim \mathcal{H}_{0}$. The spaces $\mathcal{F}%
_{k}$ and $\mathcal{H}_{k},\,k=0,1$, and their projection operators are
defined in Sec. \ref{noninv}. The mapping (\ref{r35}) is the $\Lambda $%
-extension of the operator $T(R)$ on $\mathcal{A}(\mathcal{H})$ that is
constructed from $T_{0}$ and $T_{1}$ by
\begin{equation}
T(R)\left( F_{0}\wedge F_{1}\right) =\left( T_{0}(R_{0})\,F_{0}\right)
\wedge \left( T_{1}(R_{1})\Gamma ((-1)^{n}Q_{1})\,F_{1}\right)  \label{r36}
\end{equation}%
with tensors $F_{k}\in \mathcal{A}(\mathcal{F}_{k}),k=0,1$. From the Sects. %
\ref{invertible} and \ref{dual} we know that the operators $T_{k},\,k=0,1$,
are surjective isometric mappings from $\mathcal{A}(\mathcal{F}_{k})$ onto $%
\mathcal{A}(\mathcal{H}_{k}).$ The operator $\Gamma ((-1)^{n}Q_{1})$ is an
invertible isometry on $\mathcal{A}(\mathcal{F}_{1})$. Hence the linear
extension of (\ref{r36}) is a unitary operator on $\mathcal{A}(\mathcal{H})$%
. If we fix the phase of the operator $T_{0}(R_{0})=T_{0}\left[ P_{0}V,\hat{e%
}\right] $ as done in Sec. \ref{dual}, the definition (\ref{r35}) and the
relations (\ref{r20}) and (\ref{r23}) imply the identity%
\begin{equation}
T\left[ U,V\right] \Gamma (S)=T\left[ US,V\bar{S}\right]  \label{r39}
\end{equation}%
without additional phase factor for all $S\in \mathcal{U}(\mathcal{H})$.

In the last step we give a proof that the operators (\ref{r36}) satisfy the
intertwining relation (\ref{can7}) with the Weyl operators. Let $\Xi=\Xi
_{0}\circ\Xi_{1}\in\mathcal{A}^{\Lambda}(\mathcal{H})$ be the product of
tensors $\Xi_{k}\in\mathcal{A}^{\Lambda}(\mathcal{F}_{k}),\,k=0,1$, where $%
\Xi_{0}$ has the parity $\pi(\Xi_{0})=p$. Then the identities (\ref{w10})
and (\ref{r35}) imply $\hat{T}(R)W(\xi)\left( \Xi_{0}\circ\Xi_{1}\right)
=\left( \hat{T}_{0}(R_{0})W(Q_{0}\xi)\Xi_{0}\right) \circ\left( \hat{T}%
_{1}(R_{1})\hat{\Gamma}((-1)^{n}Q_{1})W((-1)^{p}Q_{1}\xi)\Xi_{1}\right) $.
From (\ref{w8}) we know $\hat{\Gamma}((-1)^{n}Q_{1})W((-1)^{p}Q_{1}\xi)\Xi
_{1}=W((-1)^{p+n}Q_{1}\xi)\hat{\Gamma}((-1)^{n}Q_{1})\Xi_{1}$. Using the
relations (\ref{r22}) and (\ref{r34}) we obtain%
\begin{equation}
\begin{array}{l}
\hat{T}(R)W(\xi)\left( \Xi_{0}\circ\Xi_{1}\right) = \\
\left( W(R_{0}\xi)\hat{T}_{0}(R_{0})\Xi_{0}\right) \circ\left(
W((-1)^{p+n}R_{1}\xi)\hat{T}_{1}(R_{1})\hat{\Gamma}((-1)^{n}Q_{1})\Xi
_{1}\right) .%
\end{array}
\label{r37}
\end{equation}
On the other hand we have from (\ref{w10}) and (\ref{r35})%
\begin{equation}
\begin{array}{l}
W(R\xi)\hat{T}(R)\left( \Xi_{0}\circ\Xi_{1}\right) = \\
\left( W(R_{0}\xi)\hat{T}_{0}(R_{0})\Xi_{0}\right) \circ\left(
W((-1)^{q}R_{1}\xi)\hat{T}_{1}(R_{1})\hat{\Gamma}((-1)^{n}Q_{1})\Xi
_{1}\right) .%
\end{array}
\label{r38}
\end{equation}
where $q$ is the parity of $\hat{T}_{0}(R_{0})\Xi_{0}$. If $n$ is even, the
operator $\hat{T}_{0}$ does not change the parity, and $q=p$ follows; if $n$
is odd, the operator $\hat{T}_{0}$ changes the parity, and $q=p+1\,mod\,2$
follows. In both cases we have $q=p+n\;mod\,2$, and the tensors (\ref{r37})
and (\ref{r38}) agree. Hence the relation (\ref{can7}) is derived for all $%
R\in\mathcal{O}_{2}(\mathcal{H}_{\mathbb{R}})$ and all $\xi\in\mathcal{H}%
_{\Lambda}$.

\subsection{The orbit of the vacuum\label{vacuum}}

Starting from $\Xi =\kappa _{0}\otimes 1_{vac}$ in formula (\ref{r35}) we
obtain the orbit of the vacuum in the module space $\mathcal{A}^{\Lambda }(%
\mathcal{H})$. If $U$ is invertible, we can take the definition (\ref{r11})
with $\xi =0$. The pull-back of this formula to $\mathcal{A}(\mathcal{H})$
is easily seen as%
\begin{equation}
\Phi \left[ U,V\right] =\Theta (X):=\left( \det \left( I+X^{\dagger
}X\right) \right) ^{-1/4}\exp \Omega \left( X\right)  \label{r41}
\end{equation}%
with $X=V\bar{U}^{-1}$. If $\ker U^{\dagger }=\mathcal{H}_{0}\neq \left\{
0\right\} $ the additional factor $T_{0}\left[ P_{0}V\right] \,1_{vac}=\hat{e%
}=e_{1}\wedge ...\wedge e_{n}\in \mathcal{H}_{0}^{\wedge n}$ appears that is
determined -- except for a phase factor -- by the subspace $\mathcal{H}%
_{0}=P_{0}\mathcal{H}$
\begin{equation}
\Phi \left[ U,V\right] =e_{1}\wedge ...\wedge e_{n}\wedge \Theta (X)
\label{r42}
\end{equation}%
with $X=V\bar{U}^{(-1)}$. Since $\mathrm{ran}$\thinspace $X\subset \mathcal{H%
}_{1}=\mathrm{ran}$\thinspace $U$ the tensor $\Theta (X)$ is an element of $%
\mathcal{A}(\mathcal{H}_{1})$. The vectors (\ref{r41}) have a positive
overlap with the vacuum $\left( 1_{vac}\mid \Theta (X)\right) =\left( \det
\left( I+XX^{\dagger }\right) \right) ^{-1/4}>0$, whereas the vectors (\ref%
{r42}) are always orthogonal to the vacuum.

If $U$ is invertible, the relation
\begin{equation}
\Phi\left[ U,V\right] =\Phi\left[ US,V\bar{S}\right] ,\;S\in \mathcal{U}(%
\mathcal{H}),  \label{r43}
\end{equation}
follows for the vectors (\ref{r41}). If $\ker U^{\dagger}=\mathcal{H}%
_{0}\neq\left\{ 0\right\} $, we can fix the phase of the operators $%
T_{0}(R_{0})=T_{0}\left[ P_{0}V,\hat{e}\right] $ as indicated in Sec. \ref%
{dual}. Then the rule (\ref{r39}) implies the relation (\ref{r43}) also for
the vectors (\ref{r42}). As a consequence of this phase convention the
vectors $\Phi$ only depend on the variables of the coset space $\mathcal{O}%
_{2}/\mathcal{Q}$. Using the notations $T(R)$ and $\Phi(R)$ instead of $T%
\left[ U,V\right] $ and $\Phi\left[ U,V\right] $, we obtain from (\ref{r3})
the transformation rule for vectors on the orbit of the vacuum
\begin{equation}
T(R_{2})\Phi(R_{1})=\chi\,\Phi(R_{2}R_{1}),  \label{r44}
\end{equation}
where $R_{j}\in\mathcal{O}_{2},\,j=1,2$, are orthogonal transformations, and
$\chi$ is a phase factor. If $R_{1}$ and $R_{3}=R_{2}R_{1}$ are orthogonal
transformations with invertible $U$, the transformation rule (\ref{r44}))
gets the form
\begin{equation}
T\left[ U_{2},V_{2}\right] \,\Theta(X_{1})=\chi\,\Theta(X_{3})  \label{r45}
\end{equation}
with $X_{3}$ calculated from the product rule (\ref{orth2}). We can choose
the operators $U_{1}=\left( I+X_{1}X_{1}^{+}\right) ^{-\frac{1}{2}}$ and $%
V_{1}=X_{1}\left( I+X_{1}^{+}X_{1}\right) ^{-\frac{1}{2}}$, cf. Remark \ref%
{lift} in Sec. \ref{quot}. Then $U_{3}=U_{2}U_{1}+V_{2}\bar{V}_{1}$ and $%
V_{3}=U_{2}V_{1}+V_{2}\bar{U}_{1}$ yield the operator
\begin{equation}
X_{3}=V_{3}\bar{U}_{3}^{-1}=\left( U_{2}X_{1}+V_{2}\right) \left( \bar {U}%
_{2}+\bar{V}_{2}X_{1}\right) ^{-1}  \label{r46}
\end{equation}
on the right side of (\ref{r45}). A formula of this type is given in Sec.
12.2 of \cite{Pressley/Segal:1986} for the finite dimensional orthogonal
group.

\appendix

\section{Fock space calculations\label{calc}}

\subsection{Products\label{prod}}

The following norm estimate for the exterior algebra is well known%
\begin{equation}
\left\Vert F\wedge G\right\Vert \leq \sqrt{\frac{(p+q)!}{p!q!}}\left\Vert
F\right\Vert \left\Vert G\right\Vert \;\mathrm{if}\;F\in \mathcal{A}_{p}(%
\mathcal{H}),\,G\in \mathcal{A}_{q}(\mathcal{H}),  \label{p1}
\end{equation}%
where $\left\Vert .\right\Vert $ is the standard Fock space norm (\ref{h3})
and $p,q\in \left\{ 0\right\} \cup \mathbb{N}$ are the degrees of the
tensors. The modified norm (\ref{h4}) of the Grassmann algebra $\Lambda $
leads to the estimates%
\begin{equation}
\left\Vert \lambda \mu \right\Vert _{\Lambda }\leq \sqrt{\frac{p!q!}{(p+q)!}}%
\left\Vert \lambda \right\Vert _{\Lambda }\left\Vert \mu \right\Vert
_{\Lambda }\leq \left\Vert \lambda \right\Vert _{\Lambda }\left\Vert \mu
\right\Vert _{\Lambda }\;\mathrm{if}\;\lambda \in \Lambda _{p},\,\mu \in
\Lambda _{q}.  \label{p2}
\end{equation}%
These stricter bounds imply that the Grassmann product is continuous with
\cite{Kupsch/Smolyanov:1998, Kupsch/Smolyanov:2000}
\begin{equation}
\left\Vert \lambda _{1}\lambda _{2}\right\Vert _{\Lambda }\leq \sqrt{3}%
\left\Vert \lambda _{1}\right\Vert _{\Lambda }\left\Vert \lambda
_{2}\right\Vert _{\Lambda }\;\mathrm{for\,all}\;\lambda _{1,2}\in \Lambda .
\label{p3}
\end{equation}%
If $\lambda \in \Lambda _{1}$ the estimate $\left\Vert \lambda \mu
\right\Vert _{\Lambda }\leq \left\Vert \lambda \right\Vert _{\Lambda
}\left\Vert \mu \right\Vert _{\Lambda }$ follows from (\ref{p2}) for all $%
\mu \in \Lambda $. The definition (\ref{h4}) of the norm implies that the
product $\lambda _{\mathbf{K}}$ of vectors in the generating space $\lambda
_{k}\in \Lambda _{1},\,k\in \mathbf{K},\,\left\vert \mathbf{K}\right\vert
=p\in \mathbb{N}$, has a norm $\left\Vert \lambda _{\mathbf{K}}\right\Vert
_{\Lambda }\leq \left( p!\right) ^{-1}{\prod\limits_{k\in \mathbf{K}}}%
\left\Vert \lambda _{k}\right\Vert _{\Lambda }$.

The product $\Theta \circ \Xi $ is not defined on the full module Fock space
$\mathcal{A}^{\Lambda }(\mathcal{H})$, but it can be extended to a larger
space than $\mathcal{A}_{fin}(\mathcal{H)}$ by continuity arguments. Any
tensor $\Theta \in \mathcal{A}_{fin}^{\Lambda }(\mathcal{H})$ can be
decomposed into $\Theta =\sum_{p=0}^{\infty }\Theta _{p}$ with $\Theta \in
\Lambda \widehat{\otimes }\mathcal{A}_{p}(\mathcal{H})$. The product has the
structure $\Theta _{p}\circ \Xi _{q}\in \Lambda \widehat{\otimes }\mathcal{A}%
_{p+q}(\mathcal{H})$ if $\Theta _{p}\in \Lambda \widehat{\otimes }\mathcal{A}%
_{p}(\mathcal{H})$ and $\Xi _{q}\in \Lambda \widehat{\otimes }\mathcal{A}%
_{q}(\mathcal{H})$ with norm estimate
\begin{equation}
\left\Vert \Theta _{p}\circ \Xi _{q}\right\Vert _{\otimes }\leq \sqrt{3\,%
\frac{(p+q)!}{p!q!}}\left\Vert \Theta _{p}\right\Vert _{\otimes }\left\Vert
\Xi _{q}\right\Vert _{\otimes }.  \label{p4}
\end{equation}%
We introduce the family of Hilbert norms $\left\Vert \Theta \right\Vert
_{(\alpha )}^{2}=\sum_{p=0}^{\infty }(p!)^{\alpha }\left\Vert \Theta
_{p}\right\Vert _{\otimes }^{2}$ with a parameter $\alpha \geq 0$. The
completion of $\mathcal{A}_{fin}^{\Lambda }(\mathcal{H})$ with respect to
the norm $\left\Vert .\right\Vert _{(\alpha )}$ is called
$\mathcal{A}_{(\alpha )}^{\Lambda }(\mathcal{H})$. The
inclusions $\mathcal{A}_{fin}^{\Lambda }(\mathcal{H})\subset \mathcal{A}%
_{(\alpha )}^{\Lambda }(\mathcal{H})\subset \mathcal{A}_{(0)}^{\Lambda }(%
\mathcal{H})=\mathcal{A}^{\Lambda }(\mathcal{H}),\,\alpha \geq 0$, are
obvious. Then Proposition 3 in Appendix A of \cite{Kupsch/Smolyanov:1998}
(or the results of \cite{Kupsch/Smolyanov:2000}) imply the statement:

\begin{lemma}
\label{product1}The product of $\mathcal{A}_{fin}^{\Lambda }(\mathcal{H})$
can be extended to a continuous mapping $\mathcal{A}_{(\alpha )}^{\Lambda }(%
\mathcal{H})\times \mathcal{A}_{(\beta )}^{\Lambda }(\mathcal{H}%
)\longrightarrow \mathcal{A}^{\Lambda }(\mathcal{H})$ for spaces $\mathcal{A}%
_{(\alpha )}^{\Lambda }(\mathcal{H})$ and $\mathcal{A}_{(\beta )}^{\Lambda }(%
\mathcal{H})$ if $\alpha >0$ and $\beta >0$.
\end{lemma}

Another extension of the product is formulated in \textbf{Lemma \ref%
{product2}}, for which the proof is given here. We start with a tensor $%
\Theta \in \mathcal{A}_{fin}^{\Lambda }(\mathcal{H})$ and choose an ON basis
$\left\{ \kappa _{m}\mid m\in \mathbb{N}\right\} $ of $\Lambda _{1}$. Then $%
\xi \in \mathcal{H}_{\Lambda }$ can be written as $\xi =\sum_{m\in \mathbf{M}%
}\kappa _{m}\otimes h_{m}$ with $h_{m}\in \mathcal{H}$\thinspace ,\thinspace
$\sum_{m\in \mathbf{M}}\left\Vert h_{m}\right\Vert ^{2}=\left\Vert \xi
\right\Vert _{\otimes }^{2}$ and an index set $\mathbf{M}\subset \mathbb{N}$%
. The set $\left\{ p!\kappa _{\mathbf{L}}\mid \mathbf{L}\in \mathcal{P}(%
\mathbb{N}),\,p=\left\vert \mathbf{L}\right\vert \right\} $ is an ON basis
of $\Lambda $, and the tensor $\Theta $ can be expanded as $\Theta =\sum_{%
\mathbf{L}\in \mathcal{P}(\mathbb{N})}p!\kappa _{\mathbf{L}}\otimes F(%
\mathbf{L})$ with $F(\mathbf{L})\in \mathcal{A}(\mathcal{H}),\,\sum_{\mathbf{%
L}\in \mathcal{P}(\mathbb{N})}\left\Vert F(\mathbf{L})\right\Vert
^{2}=\left\Vert \Theta \right\Vert _{\otimes }^{2}$. The tensor $\xi \circ
\Theta =\sum_{m,\mathbf{L}}p!\kappa _{m}\kappa _{\mathbf{L}}\otimes \left(
h_{m}\wedge F(\mathbf{L})\right) $ has the norm \newline
$\left\Vert \xi \circ \Theta \right\Vert _{\otimes }^{2}\leq \sum_{\mathbf{K}%
\in \mathcal{P}(\mathbb{N})}\left( (p+1)^{-1}\sum^{\prime }\left\Vert
h_{m}\wedge F(\mathbf{L})\right\Vert \right) ^{2}$, where the sum $%
\sum^{\prime }$ extends over the $p+1$ pairs $\left\{ \left( m,\mathbf{L}%
\right) \in \mathbb{N}\times \mathcal{P}(\mathbb{N})\mid \left\{ m\right\}
\cup \mathbf{L}=\mathbf{K},\,\left\vert \mathbf{K}\right\vert =p+1\right\} $%
. This inequality implies\footnote{%
Here we use the estimate $\left( \sum_{j=1}^{n}x_{j}\right) ^{2}\leq n\left(
\sum_{j=1}^{n}x_{j}^{2}\right) $ if $x_{j}\geq 0,\,j=1,...,n$.} \newline
$\left\Vert \xi \circ \Theta \right\Vert _{\otimes }^{2}\leq \sum_{m,\mathbf{%
N}}\left\Vert h_{m}\wedge F(\mathbf{N})\right\Vert ^{2}\leq \left(
\sum_{n\in \mathbb{N}}\left\Vert h_{n}\right\Vert ^{2}\right) \cdot \sum_{%
\mathbf{L}\in \mathcal{P}(\mathbb{N})}\left\Vert F(\mathbf{L})\right\Vert
^{2}=\left\Vert \xi \right\Vert _{\otimes }^{2}\left\Vert \Theta \right\Vert
_{\otimes }^{2}$. Hence the operator norm estimate (\ref{h5}) is true for $%
\Theta \in \mathcal{A}_{fin}^{\Lambda }(\mathcal{H})$. The extension to
tensors $\Theta \in \mathcal{A}^{\Lambda }(\mathcal{H})$ follows by
continuity.

The inequality (\ref{h5}) implies that the creation operator $b^{+}(\xi)$ is
a continuous operator on $\mathcal{A}^{\Lambda}(\mathcal{H})$ with operator
norm $\left\Vert b^{+}(\xi)\right\Vert \leq\left\Vert \xi\right\Vert
_{\otimes}$. Moreover, a slight modification of the proof yields that the
annihilation operator $b^{-}(\xi)=\sum_{m\in\mathbf{M}}\kappa_{m}^{\ast}%
\otimes a^{-}(h_{m})$ is also continuous with operator norm $\left\Vert
b^{-}(\xi)\right\Vert \leq\left\Vert \xi^{\ast}\right\Vert
_{\otimes}=\left\Vert \xi\right\Vert _{\otimes}$.

\subsection{Exponentials of tensors of second degree\label{exponentials}}

If $X$ is an operator in $\mathcal{L}_{2}^{-}(\mathcal{H})$, the following
three statements are valid:

\begin{enumerate}
\item There exists an ON system $\left\{ e_{m}\in\mathcal{H}\mid m\in\mathbf{%
M}\cup(-\mathbf{M})\right\} $ with $\mathbf{M}=\left\{ 1,2,...,M\right\} $
or $\mathbf{M}=\mathbb{N}$ and complex numbers $z_{m}\in\mathbb{C}$ with $%
\left\vert z_{1}\right\vert \geq\left\vert z_{2}\right\vert \geq...>0$ so
that the mapping $\mathcal{H}\ni f\rightarrow Xf\in\mathcal{H}$ has the
representation
\begin{equation}
Xf=\sum_{m\in\mathbf{M}}z_{m}\left( e_{m}\left\langle e_{-m}\mid
f\right\rangle -e_{-m}\left\langle e_{m}\mid f\right\rangle \right) .
\label{g3}
\end{equation}
The numbers $z_{m}$ are square-summable with $\sum_{m\in\mathbf{M}%
}\left\vert z_{m}\right\vert ^{2}=2^{-1}\mathrm{tr}X^{\dag}X=2^{-1}\left%
\Vert X\right\Vert _{2}^{2}$. Here $\left\langle f\mid g\right\rangle
=\left( f^{\ast}\mid g\right) $ is the $\mathbb{C}$-bilinear symmetric form
introduced in Sec. \ref{hilbert}.

\item There exists exactly one tensor $\Omega(X)\in\mathcal{A}_{2}(\mathcal{H%
})$ such that the identities%
\begin{equation}
\left\langle \Omega(X)\mid f\wedge g\right\rangle =\left\langle f\wedge
g\mid\Omega(X)\right\rangle =\left\langle f\mid Xg\right\rangle
=-\left\langle g\mid Xf\right\rangle  \label{g4}
\end{equation}
are true for all $f,g\in\mathcal{H}$. The mapping $\mathcal{L}_{2}^{-}(%
\mathcal{H})\ni X\rightarrow\Omega(X)\in\mathcal{A}_{2}(\mathcal{H})$ is a
linear and continuous isomorphism, the norms are related by $\left\Vert
\Omega(X)\right\Vert _{2}^{2}=2^{-1}\left\Vert X\right\Vert _{HS}^{2}$.

\item The exponential series (\ref{h11}) converges absolutely, and the inner
product of two of these exponentials is (\ref{h15}).
\end{enumerate}

The first Statement is an immediate consequence of Lemma 4.1 in \cite%
{Ruijsenaars:1978}.

The identities (\ref{g4}) imply that $X\rightarrow \Omega (X)$ is a linear
bijective mapping from $X\in \mathcal{F}^{-}(\mathcal{H})$ -- the space of
skew symmetric finite rank operators -- onto $\Omega (X)\in \mathcal{H}%
\wedge \mathcal{H}$. Given $X\in \mathcal{F}^{-}(\mathcal{H})$ the operator $%
X$ has the representation (\ref{g3}) with a finite index set $\mathbf{M}$.
The tensor
\begin{equation}
\Omega (X):=\sum_{m\in \mathbf{M}}z_{m}\,e_{-m}\wedge e_{m}  \label{g6}
\end{equation}%
is an element of $\mathcal{H}\wedge \mathcal{H}$ and it satisfies the
identities (\ref{g4}). Then the norm identity $\left\Vert X\right\Vert
_{2}^{2}=2\sum_{a\in \mathbf{M}}\left\vert z_{m}\right\vert ^{2}=2\left\Vert
\Omega (X)\right\Vert _{2}^{2}$ yields that the linear mapping $\mathcal{F}%
^{-}(\mathcal{H})\ni X\rightarrow \Omega (X)\in \mathcal{A}_{2}(\mathcal{H})$
can be extended by continuity to $\mathcal{L}_{2}^{-}(\mathcal{H})$. The
relation (\ref{g4}) is equivalent to \newline
$\left( f\wedge g\mid \Omega (X)\right) =\left( g\mid X\,f^{\ast }\right) $.

For the proof of Statement 3 we start from the representation (\ref{g6}) of
the tensor $\Omega (X)$ with a finite or countable set $\mathbf{M}\subset
\mathbb{N}$. Then the powers of $\Omega $ are calculated as $\Omega ^{\wedge
p}(X)=p!\sum_{\mathbf{A},\left\vert \mathbf{A}\right\vert =p}z_{\mathbf{A}%
}\,e_{-\mathbf{A}}\wedge e_{\mathbf{A}}$. Thereby $\sum_{\mathbf{A}%
,\left\vert \mathbf{A}\right\vert =p}$ means summation over all subsets $%
\mathbf{A}\subset \mathbf{M}$ with $\left\vert \mathbf{A}\right\vert =p\geq
1 $ elements. The exponential series is $\exp \Omega (X)=\sum_{\mathbf{A}}z_{%
\mathbf{A}}\,e_{-\mathbf{A}}\wedge e_{\mathbf{A}}$, where the sum extends
over all finite subsets $\mathbf{A}$ of $\mathbf{M}$, including $\mathbf{A}%
=\emptyset $. The inner product%
\begin{equation}
\left( \Omega ^{\wedge p}(X)\mid \Omega ^{\wedge p}(X)\right) =(p!)^{2}\sum_{%
\mathbf{A},\left\vert \mathbf{A}\right\vert =p}\left\vert
z_{a_{1}}\right\vert ^{2}...\left\vert z_{a_{p}}\right\vert ^{2}\leq
p!\left( \sum_{n\in \mathbf{M}}\left\vert z_{n}\right\vert ^{2}\right)
^{p}=p!\left( \frac{1}{2}\left\Vert X\right\Vert _{2}^{2}\right) ^{p}
\label{g9}
\end{equation}%
yields the identities%
\begin{equation}
\begin{array}{l}
\left\Vert \exp \Omega (X)\right\Vert ^{2}=\sum_{p=0,1,2,...}\left(
\sum_{n_{1}<n_{2}<...<n_{p}}\left\vert z_{n_{1}}\right\vert
^{2}...\left\vert z_{n_{p}}\right\vert ^{2}\right) =\prod\limits_{n\in
\mathbf{M}}(1+\left\vert z_{n}\right\vert ^{2})= \\
\sqrt{\det \left( I+X^{\dag }X\right) }=\sqrt{\det \left( I+XX^{\dag
}\right) }\leq \exp \left( \frac{1}{2}\left\Vert X\right\Vert
_{2}^{2}\right) .%
\end{array}
\label{g10}
\end{equation}%
Hence the exponential series (\ref{h11}) converges in norm within $\mathcal{A%
}(\mathcal{H})$ for all $X\in \mathcal{L}_{2}^{-}(\mathcal{H})$, and $%
\mathcal{L}_{2}^{-}(\mathcal{H})\ni X\rightarrow \exp \Omega (X)\in \mathcal{%
A}(\mathcal{H})$ is an entire analytic function. The mapping \newline
$\mathcal{L}_{2}^{-}(\mathcal{H})\times \mathcal{L}_{2}^{-}(\mathcal{H})\ni
(X,Y)\rightarrow \varphi (X,Y):=\left( \exp \Omega (X)\mid \exp \Omega
(Y)\right) \in \mathbb{C}$ is antianalytic in $X$ and analytic in $Y.$ Hence
the function $\varphi (X,Y)$ is uniquely determined by the values on the
diagonal $X=Y$, and $\varphi (X,X)=\sqrt{\det \left( I+X^{\dag }X\right) }=%
\sqrt{\det \left( I+XX^{\dag }\right) }$ implies (\ref{h15}).

Since the left side of (\ref{h15}) is antianalytic in $X$ and analytic in $Y$
the function $\sqrt{\det(I+X^{\dag}Y)}$ can be expanded in a power series of
the variables $X^{\dagger}$ and $Y$. An explicit form can be obtained from
the left side of (\ref{h15}). This expansion is often formulated with
pfaffians of the operators $X^{\dagger}$ and $Y$ restricted to final
dimensional subspaces, cf. e.g. Sec. 12 of \cite{Pressley/Segal:1986} or
\cite{JLW:1989}.

The exponential of a tensor $\lambda \in \Lambda _{2}$ converges within $%
\Lambda $ with norm estimate \newline
$\left\Vert \exp \lambda \right\Vert _{\Lambda }^{2}\leq 1+\left\Vert
\lambda \right\Vert _{\Lambda }^{2}+\left( 2!\right) ^{-2}\left\Vert \lambda
^{2}\right\Vert _{\Lambda }^{2}+...\overset{(\ref{p2})}{\leq }%
\sum_{p=0}^{\infty }\left( p!\right) ^{-2}\left\Vert \lambda \right\Vert
_{\Lambda }^{2p}\leq \exp \left( \left\Vert \lambda \right\Vert _{\Lambda
}^{2}\right) $.

\subsection{Ultracoherent vectors\label{ultra}}

The mapping $\mathcal{A}(\mathcal{H})\ni F\rightarrow \kappa _{0}\otimes
F\in \mathcal{A}^{\Lambda }(\mathcal{H})$ gives a natural embedding of the
Fock space $\mathcal{A}(\mathcal{H})$ into the $\Lambda $-module $\mathcal{A}%
^{\Lambda }(\mathcal{H})=\Lambda \widehat{\otimes }\mathcal{A}(\mathcal{H})$%
. The tensor (\ref{h11}) can therefore be identified with the element $\Psi
(X):=\kappa _{0}\otimes \exp \Omega (X)=\exp \left( \kappa _{0}\otimes
\Omega (X)\right) $ of $\mathcal{A}^{\Lambda }(\mathcal{H})$. The norm of $%
\Psi (X)$ is $\left\Vert \Psi (X)\right\Vert _{\otimes }=\left\Vert \exp
\Omega (X)\right\Vert \overset{(\ref{g10})}{=}\left( \det \left( I+X^{\dag
}X\right) \right) ^{1/4}$. If $\xi \in \mathcal{H}_{\Lambda }^{_{alg}}$ is
an element of the algebraic superspace and $X\in \mathcal{L}_{2}^{-}(%
\mathcal{H})$ is a finite rank operator, the exponential functions are
finite sums, and the identities $\left( \exp \xi
\right) \circ \Psi (X)=\Psi (X)\circ \left( \exp \xi \right) =\exp \left(
\xi +\kappa _{0}\otimes \Omega (X)\right) $ follow by algebraic calculation.
The norm estimate (\ref{h5}) implies that the products $\left( \exp \xi
\right) \circ \Psi (X)$ and $\Psi (X)\circ \left( \exp \xi \right) $ are
defined with a norm $\left\Vert \left( \exp \xi \right) \circ \Psi
(X)\right\Vert _{\otimes }\leq \sum_{p}(p!)^{-1}\left\Vert \xi ^{p}\circ
\Psi (X)\right\Vert _{\otimes }\leq \\ \sum_{p}(p!)^{-1}\left\Vert \xi
\right\Vert _{\otimes }^{p}\cdot \left\Vert \Psi (X)\right\Vert _{\otimes
}=\left( \exp \left\Vert \xi \right\Vert _{\otimes }\right) \left\Vert \Psi
(X)\right\Vert _{\otimes }$. Hence the ultracoherent vector
\begin{equation}
\Psi (X,\xi ):=\left( \exp \xi \right) \circ \Psi (X)=\Psi (X)\circ \left(
\exp \xi \right)  \label{g21}
\end{equation}%
is a well defined element of $\mathcal{A}^{\Lambda }(\mathcal{H})$ for all $%
X\in \mathcal{L}_{2}^{-}(\mathcal{H})$ and $\,\xi \in \mathcal{H}_{\Lambda }$%
. Moreover, the inequalities (\ref{h5}) and (\ref{g9}) imply the estimate $%
\left\Vert \xi ^{p}\circ \left( \kappa _{0}\otimes \Omega ^{\wedge q}\right)
\right\Vert _{\otimes }\leq \sqrt{q!}\left\Vert \xi \right\Vert _{\otimes
}^{p}\cdot \left\Vert \Omega \right\Vert ^{q}$ for all integers $p,q\geq 0$.
Hence the series $\exp \left( \xi +\kappa _{0}\otimes \Omega (X)\right) $
converges uniformly in $\xi $ and $X$ within the space $\mathcal{A}^{\Lambda
}(\mathcal{H})$ with the norm estimate $\left\Vert \exp \left( \xi +\kappa
_{0}\otimes \Omega (X)\right) \right\Vert _{\otimes }\leq \sum_{p,q\geq
0}\left( p!\right) ^{-1}\left( q!\right) ^{-\frac{1}{2}}\left\Vert \xi
\right\Vert _{\otimes }^{p}\left( 2^{-\frac{1}{2}}\left\Vert X\right\Vert
_{2}\right) ^{q}$. The tensor (\ref{g21}) therefore coincides with $\exp
\left( \xi +\kappa _{0}\otimes \Omega (X)\right) $ for all $X\in \mathcal{L}%
_{2}^{-}(\mathcal{H})$ and $\,\xi \in \mathcal{H}_{\Lambda }$. The tensors $%
\exp \xi ,\,\Psi (X)$ and $\Psi (X,\xi )$ have even parity.

\begin{remark}
The existence of the products (\ref{g21}) can also be derived from Lemma \ref%
{product1}. The power $\xi ^{p}$ is an element of $\Lambda \widehat{\otimes }%
\mathcal{A}_{p}(\mathcal{H})$ with norm $\left\Vert \xi ^{p}\right\Vert
_{\otimes }^{2}\leq \left\Vert \xi \right\Vert _{\otimes }^{2p}$. Hence $%
\exp \xi $ is an element of the space $\mathcal{A}_{(\alpha )}^{\Lambda }(%
\mathcal{H})$ if $\alpha \in \left[ 0,2\right) $. The tensor $\kappa
_{0}\otimes \Omega ^{\wedge q}(X)$ is an element of $\Lambda \widehat{%
\otimes }\mathcal{A}_{2q}(\mathcal{H})$ with norm $\left\Vert \kappa
_{0}\otimes \Omega ^{\wedge q}(X)\right\Vert _{\otimes }^{2}\leq \left(
q!\right) \left\Vert X\right\Vert _{2}^{2q}$, cf. (\ref{g9}). Hence $\Psi
(X) $ is an element of the space $\mathcal{A}_{(\beta )}^{\Lambda }(\mathcal{%
H})$ if $\beta \in \left[ 0,1/2\right) $. The conditions of Lemma \ref%
{product1} for the product (\ref{g21}) are therefore satisfied.
\end{remark}

The relation $\left( f\wedge g\mid \Omega (X)\right) =\left( g\mid
X\,f^{\ast }\right) ,\,f,g\in \mathcal{H}$, implies $\left( \xi \circ \xi
\mid \kappa _{0}\otimes \Omega (X)\right) =\left( \xi \mid X\,\xi ^{\ast
}\right) $ with $\xi \in \mathcal{H}_{\Lambda }$. Then the identity
\begin{equation}
\left( \exp \xi \mid \Psi (X)\right) =\exp \frac{1}{2}\left( \xi \mid X\,\xi
^{\ast }\right)   \label{g22}
\end{equation}%
follows by series expansion and repeated use of the identity (\ref{h17}).
The relation (\ref{h14}) is a consequence of the identities (\ref{h16}), (%
\ref{h17}) and (\ref{g22}).

To derive the action of the Weyl operator on an ultracoherent vector $%
\Psi(X,\xi)$ we calculate with the variables $\xi,\eta\in\mathcal{H}%
_{\Lambda}$ and $X\in\mathcal{L}_{2}^{-}(\mathcal{H})$ using (\ref{w1}) and (%
\ref{h14})\newline
$\left( \exp\,\xi\mid W(\eta)\Psi(X,\xi)\right) =\left(
W^{+}(\eta)\exp\,\xi\mid\Psi(X,\xi)\right) =\newline
\exp\left( -\frac{1}{2}\left\langle \eta^{\ast}\parallel\eta\right\rangle +%
\frac{1}{2}\left\langle \eta^{\ast}\parallel X\eta^{\ast}-2\xi\right\rangle
\right) \left\langle \exp\xi^{\ast}\parallel\exp\left(
\xi+\eta-X\eta^{\ast}+\kappa_{0}\otimes\Omega(X)\right) \right\rangle $.
This identity implies the formula (\ref{w7}). The restriction of (\ref{w7})
to $\Psi(X)$ is
\begin{equation}
W(\eta)\Psi(X)=\exp\left( -\frac{1}{2}\left\langle \eta^{\ast}\parallel
\eta-X\eta\right\rangle \right) \Psi(X,\eta-X\eta^{\ast}).  \label{g23}
\end{equation}

Let $\xi$ be a supervector in $\mathcal{H}_{\Lambda}$, then $%
\eta=(I+XX^{\dag })^{-1}\xi+X(I+X^{\dag}X)^{-1}\xi^{\ast}$ is an element of $%
\mathcal{H}_{\Lambda}$, which satisfies $\eta-X\eta^{\ast}=\xi$. With this
supervector $\xi$ we obtain from (\ref{g23})
\begin{equation}
W(\eta)\Psi(X)=\Psi(X,\xi)\exp\left( \frac{1}{2}\left\langle \xi
\parallel(I-X^{\dag}X)^{-1}\xi^{\ast}+X^{\dag}(I-XX^{\dag})^{-1}\xi\right%
\rangle \right) .  \label{g24}
\end{equation}
The inner product $\left( W\,\Psi\mid W\,\Psi\right) $ is known as $\left(
W(\eta)\Psi(X)\mid W(\eta)\Psi(X)\right) \overset{(\ref{w6})}{=}\left(
\Psi(X)\mid\Psi(X)\right) \overset{(\ref{g10})}{=}\kappa_{0}\otimes\sqrt {%
\det\left( I+X^{\dag}X\right) }$. Substituting (\ref{g24}) into this
identity we obtain formula (\ref{h13}).

\section{Calculations for Sec. \ref{rep}\label{calcRep}}

\subsection{The operator $T(R)$ of Sec. \ref{invertible}\label{pull}}

In this Appendix we calculate the image of the operator $T(R)$ on arbitrary
factorizing tensors $f_{\mathbf{M}}\in \mathcal{A}_{p}(\mathcal{H)}$ with $%
\left\vert \mathbf{M}\right\vert =M\in \mathbb{N}$. We start from the
expansion $\xi =\sum_{m\in \mathbb{N}}\kappa _{m}\otimes f_{m}\in \mathcal{H}%
_{\Lambda }$ with an ON basis $\left\{ \kappa _{m}\right\} $ of $\Lambda
_{1} $ and vectors $f_{m}\in \mathcal{H}$ with $\sum_{m}\left\Vert
f_{m}\right\Vert ^{2}=\left\Vert \xi \right\Vert _{\otimes }^{2}$. The
coherent vector is $\exp \xi =\sum_{\mathbf{M}\subset \mathbb{N}}\kappa _{%
\mathbf{M}}\otimes f_{\mathbf{M}}$. The quadratic form $\left\langle \xi
\mid Y\,\xi \right\rangle $ agrees with $\left\langle \kappa _{0}\otimes
\Omega (Y)\mid \xi \circ \,\xi \right\rangle $, cf. Appendix \ref{ultra},
and its exponential is, cf. (\ref{g22}),%
\begin{equation*}
\exp \left( \frac{1}{2}\left\langle \xi \mid Y\,\xi \right\rangle \right)
=\left\langle \exp \left( \kappa _{0}\otimes \Omega (Y)\right) \mid \exp \xi
\right\rangle =\sum_{\mathbf{M}\subset \mathbb{N}}\varphi (\mathbf{M}%
)\,\kappa _{\mathbf{M}}
\end{equation*}%
with $\varphi (\mathbf{M}):=\left\langle \exp \Omega (Y)\mid f_{\mathbf{M}%
}\right\rangle $. The numbers $\varphi (\mathbf{M})$ have the values \newline
$\varphi (\emptyset )=1$, $\varphi (\mathbf{M})=0$ if $\left\vert \mathbf{M}%
\right\vert $ is odd, and
\begin{equation*}
\varphi (\mathbf{M})=(q!)^{-1}\left\langle \Omega ^{\wedge q}(Y)\mid f_{%
\mathbf{M}}\right\rangle =\mathrm{Pf}\,\left( \left\langle f_{m}\mid
Yf_{n}\right\rangle _{\mathbf{M}}\right)
\end{equation*}%
if $\left\vert \mathbf{M}\right\vert =M=2q,\,q\in \mathbb{N}$. Thereby $%
\left\langle f_{m}\mid Yf_{n}\right\rangle _{\mathbf{M}}$ is the skew
symmetric $M\times M$ matrix $\left\{ \left\langle f_{m}\mid
Yf_{n}\right\rangle ,\,m\in \mathbf{M},\,n\in \mathbf{M}\right\} $, and $%
\mathrm{Pf}$ is the pfaffian of this matrix. The $\Lambda $-dependent
factors of the right side of (\ref{r11}) are
\begin{equation*}
\begin{array}{l}
\exp \left( \frac{1}{2}\left\langle \xi \mid Y\xi \right\rangle \right) \exp
\left( U^{\dagger -1}\xi \right) = \\
\left( \sum_{\mathbf{K}\subset \mathbb{N}}\varphi (\mathbf{K})\,\kappa _{%
\mathbf{K}}\right) \left( \sum_{\mathbf{L}\subset \mathbb{N}}\kappa _{%
\mathbf{L}}\otimes (U^{\dagger -1}f)_{\mathbf{L}}\right) = \\
\sum_{\mathbf{M}\subset \mathbb{N}}\kappa _{\mathbf{M}}\otimes \left( \sum_{%
\mathbf{K}\cup \mathbf{L=M},\mathbf{K}\cap \mathbf{L}=\emptyset }(-1)^{\tau (%
\mathbf{K},\mathbf{L})}\varphi (\mathbf{K})(U^{\dagger -1}f)_{\mathbf{L}%
}\right) .%
\end{array}%
\end{equation*}%
$\newline
$The sign factor comes from $\kappa _{\mathbf{K}}\kappa _{\mathbf{L}%
}=(-1)^{\tau (\mathbf{K},\mathbf{L})}\kappa _{\mathbf{K}\cup \mathbf{L}}$.
Including the factor $c_{X}\exp \left( \kappa _{0}\otimes \Omega (X)\right) $
the right side of the ansatz (\ref{r11}) gets the form $\sum_{\mathbf{M}%
\subset \mathbb{N}}\kappa _{\mathbf{M}}\otimes F(\mathbf{M})$, where the
tensors $F(\mathbf{M})\in \mathcal{A}(\mathcal{H})$ are given by%
\begin{equation*}
F(\mathbf{M})=T(R)\,f_{\mathbf{M}}=c_{X}\left( \sum_{\mathbf{K}\cup \mathbf{%
L=M},\mathbf{K}\cap \mathbf{L}=\emptyset }(-1)^{\tau (\mathbf{K},\mathbf{L}%
)}\varphi (\mathbf{K})\,(U^{\dagger -1}f)_{\mathbf{L}}\right) \wedge \exp
\Omega (X).
\end{equation*}%
For tensors of degree less than 3 we obtain $T(R)1_{vac}=c_{X}\exp \Omega
(X),\newline
T(R)f=c_{X}\,(U^{\dagger -1}f)\wedge \exp \Omega (X)$, and\newline
$T(R)\left( f_{1}\wedge f_{2}\right) =c_{X}\,\left( (U^{\dagger
-1}f_{1})\wedge (U^{\dagger -1}f_{2}))+\left\langle f_{1}\mid
Yf_{2}\right\rangle 1_{vac}\right) \wedge \exp \Omega (X).$

\subsection{Calculations for the duality mapping\label{reflections}}

If $\mathbf{A}$ and $\mathbf{B}$ are finite subsets of $\mathbb{N}$ and $n$
is an element of $\mathbb{N}$, we define the following numbers: $\tau (n,%
\mathbf{B}):=\#\left\{ b\in\mathbf{B}\;\mathrm{with}\;n>b\right\} $ and $%
\tau(\mathbf{A},\mathbf{B}):=\#\left\{ (a,b)\in\mathbf{A}\times \mathbf{B}%
\mid a>b\right\} $. Let $\mathcal{H}$ be a Hilbert space with ON basis $%
\left\{ e_{n}\mid n\in\mathbb{N}\right\} $. The creation/annihilation
operators $a_{n}^{\pm}=a^{\pm}(e_{n})$ are determined by their values on the
ON basis vectors $\left\{ e_{\mathbf{A}}\mid\mathbf{A}\subset\mathbb{N}%
\right\} $ of $\mathcal{A}(\mathcal{H})$:%
\begin{equation}
\begin{array}{lll}
a_{n}^{+}e_{\mathbf{A}}=0\;\mathrm{if}\;n\in\mathbf{A}, &  & a_{n}^{+}e_{%
\mathbf{A}}=(-1)^{\tau(n,\mathbf{A})}e_{\mathbf{A}\cup\left\{ n\right\} }\;%
\mathrm{if}\;n\notin\mathbf{A}, \\
a_{n}^{-}e_{\mathbf{A}}=(-1)^{\tau(n,\mathbf{A})}e_{\mathbf{A}\backslash
\left\{ n\right\} }\;\mathrm{if}\;n\in\mathbf{A}, &  & a_{n}^{-}e_{\mathbf{A}%
}=0\;\mathrm{if}\;n\notin\mathbf{A}.%
\end{array}
\label{c1}
\end{equation}

Taking the definitions of spaces and operators from Sec. \ref{dual} the
operator $\Theta\equiv T_{0}\left[ J\right] $ is defined as%
\begin{equation}
\Theta\,f_{\mathbf{K}}:=(-1)^{\tau(\mathbf{K},\mathbf{M})}e_{\mathbf{\bar{K}}%
},\;\mathbf{K}\subset\mathbf{M}.  \label{c2}
\end{equation}
Thereby $\mathbf{M}$ is the finite set $\mathbf{M}=\left\{ 1,...,n\right\} $%
, and $\mathbf{\bar{K}}:=\mathbf{M}\backslash\mathbf{K}$ is the complement
of $\mathbf{K\subset M}$. The ON basis systems $\left\{ e_{m},\,m\in \mathbf{%
M}\right\} \subset\mathcal{H}_{0}$ and $\left\{ f_{m},\,m\in \mathbf{M}%
\right\} \subset\mathcal{F}_{0}$ are related by the involution $\mathcal{F}%
_{0}\rightarrow\mathcal{F}_{0}^{\ast}$ and by the linear isometry $J:%
\mathcal{F}_{0}^{\ast}\rightarrow\mathcal{H}_{0}$%
\begin{equation}
e_{m}=-Jf_{m}^{\ast},\;m\in\mathbf{M}.  \label{c3}
\end{equation}
Using the relations (\ref{c1}) and (\ref{c2}) we obtain the identities%
\begin{align*}
\Theta a^{+}(f_{m})f_{\mathbf{K}} & =(-1)^{\tau(m,\mathbf{K})}\Theta f_{%
\mathbf{K\cup}\left\{ m\right\} }=(-1)^{\tau(m,\mathbf{K})+\tau (\mathbf{K+}%
\left\{ m\right\} ,\mathbf{M})}e_{\left( \mathbf{\bar{K}}\backslash\left\{
m\right\} \right) }\;\mathrm{if}\;m\in\mathbf{\bar{K}} \\
\Theta a^{+}(f_{m})f_{\mathbf{K}} & =0\;\mathrm{if}\;m\in\mathbf{K} \\
a^{-}(e_{m})\Theta f_{\mathbf{K}} & =(-1)^{\tau(\mathbf{K},\mathbf{M}%
)}a^{-}(e_{m})e_{\mathbf{\bar{K}}}=(-1)^{\tau(\mathbf{K},\mathbf{M})+\tau(m,%
\mathbf{\bar{K}})}e_{\left( \mathbf{\bar{K}}\backslash\left\{ m\right\}
\right) }\;\mathrm{if}\;m\in\mathbf{\bar{K}} \\
a^{-}(e_{m})\Theta f_{\mathbf{K}} & =(-1)^{\tau(\mathbf{K},\mathbf{M}%
)}a^{-}(e_{m})e_{\mathbf{\bar{K}}}=0\;\mathrm{if}\;m\in\mathbf{K}
\end{align*}
Then the relation $\tau(m,\mathbf{K})+\tau(\mathbf{K+}\left\{ m\right\} ,%
\mathbf{M})=\tau(\mathbf{K},\mathbf{M})+\tau(m,\mathbf{\bar{K}})\;mod\,2$
implies $\Theta a^{+}(f_{m})=a^{-}(e_{m})\Theta$ for all $m\in\mathbf{M.}$
The mapping (\ref{c2}) has therefore properties
\begin{equation}
\Theta a^{\pm}(f_{m})\Theta^{\dagger}=a^{\mp}(e_{m}),\;m\in\mathbf{M}.
\label{c4}
\end{equation}
For $h\in\mathcal{F}_{0}$ a basis expansion $h=\sum_{m\in\mathbf{M}}\gamma
_{m}f_{m},\,\gamma_{m}\in\mathbb{C}$, leads to $a^{+}(h)=\sum_{m}\gamma
_{m}a^{+}(f_{m})$ and $a^{-}(h)=\sum_{m}\bar{\gamma}_{m}a^{-}(f_{m})$.
Taking into account the relations (\ref{c3}) the identities (\ref{c4}) get
the more abstract form%
\begin{equation}
\Theta a^{\pm}(h)\Theta^{\dagger}=-a^{\mp}(J\,h^{\ast})\;\mathrm{if}\;h\in%
\mathcal{F}_{0}.
\end{equation}
These identities are equivalent to (\ref{r33}).


\newpage

\begin{thebibliography}{10}

\bibitem{Araki:1968}
H.~Araki.
\newblock {On the diagonalization of a bilinear Hamiltonian by a Bogoliubov
  transformation}.
\newblock {\em Publ. RIMS, Kyoto Univ.}, 4:387--412, 1968.

\bibitem{Araki:1970}
H.~Araki.
\newblock {On quasifree states of CAR and Bogoliubov automorphisms}.
\newblock {\em Publ. RIMS, Kyoto Univ.}, 6:385--442, 1970/71.

\bibitem{Araki:1987}
H.~Araki.
\newblock Bogoliubov automorphisms and {Fock} representations of canonical
  anticommutation relations.
\newblock In P.~E.~T. Jorgensen and P.~S. Muhly, editors, {\em Contemporary
  Mathematics, Vol. 62: Operator Algebras and Mathematical Physics}, pages
  23--141. AMS, Providence, 1987.
\newblock Proceedings of a Summer Conference, June 17--21, University of Iowa
  1985.

\bibitem{BSZ:1992}
J.~C. Baez, I.~E. Segal, and Z.~Zhou.
\newblock {\em {Introduction to Algebraic and Constructive Quantum Field
  Theory}}.
\newblock Princeton University Press, Princeton, 1992.

\bibitem{BMV:1968}
E.~Balslev, J.~Manuceau, and A.~Verbeure.
\newblock {Representations of anticommutation relations and Bogoliubov
  transformations}.
\newblock {\em Commun. Math. Phys.}, 8:315--326, 1968.

\bibitem{Ben-Israel:2003}
A.~Ben-Israel and T.~N.~E. Greville.
\newblock {\em Generalized Inverses: Theory and Applications}.
\newblock Springer, New York [u.a.], 2. ed. edition, 2003.

\bibitem{Berezin:1966}
F.~A. Berezin.
\newblock {\em {The Method of Second Quantization}}.
\newblock Academic Press, New York, 1966.

\bibitem{Berezin:1987}
F.~A. Berezin.
\newblock {\em {Introduction to Superanalysis}}.
\newblock Reidel, Dordrecht, 1987.
\newblock English translation edited and revised by D. Leites.

\bibitem{Bog:1958}
N.~N. Bogoliubov.
\newblock On a new method in the theory of superconductivity.
\newblock {\em Nuovo Cimento}, 7(6):791--805, 1958.

\bibitem{BTS:1958}
N.~N. Bogoliubov, V.~V. Tolmachov, and D.~V. Shirkov.
\newblock A new method in the theory of superconductivity.
\newblock {\em Fortschr. Physik}, 6:605--682, 1958.

\bibitem{DeWitt:1992}
B.~DeWitt.
\newblock {\em Supermanifolds}.
\newblock {Cambridge Monographs on Mathematical Physics}. CUP, Cambridge, 2nd
  edition, 1992.

\bibitem{Friedrichs:1953}
K.~O. Friedrichs.
\newblock {\em {Mathemathical Aspects of the Quantum Theory of Fields}}.
\newblock Interscience, New York, 1953.

\bibitem{GGK:1990}
I.~C. Gohberg, S.~Goldberg, and M.~A. Kaashoek.
\newblock {\em Classes of linear operators}, volume~1 of {\em Operator theory.
  Vol. 49}.
\newblock Birkh\"auser, Basel, Berlin, 1990.

\bibitem{Groetsch:1977}
C.~W. Groetsch.
\newblock {\em Generalized inverses of linear operators}.
\newblock Number~37 in Pure and applied mathematics : a series of monographs
  and textbooks ; 37 ; Pure and applied mathematics. Dekker, New York, 1977.

\bibitem{JP:1981}
A.~Jadczyk and K.~Pilch.
\newblock Superspaces and supersymmetries.
\newblock {\em Commun. Math. Phys.}, 78:373--390, 1981.

\bibitem{JLW:1989}
A.~Jaffe, A.~Lesniewski, and J.~Weitsman.
\newblock {Pfaffians on Hilbert space}.
\newblock {\em J. Funct. Anal.}, 83:348--363, 1989.

\bibitem{Khrennikov:1999}
A.~Khrennikov.
\newblock {\em Superanalysis}.
\newblock Mathematics and Its Application, Vol. 470. Kluwer, Dordrecht, 1999.

\bibitem{Kupsch/Banerjee:2006}
J.~Kupsch and S.~Banerjee.
\newblock Ultracoherence and canonical transformations.
\newblock {\em Infin. Dimens. Anal. Quantum Probab. Relat. Top.},
  9(3):413--434, 2006.
\newblock Extended preprint arXiv:math-ph/0410049v3.

\bibitem{Kupsch/Smolyanov:1998}
J.~Kupsch and O.~G. Smolyanov.
\newblock {Functional representations for Fock superalgebras}.
\newblock {\em {Infin. Dimens. Anal. Quantum Probab. Relat. Top.}},
  1(2):285--324, 1998.
\newblock arXiv:hep-th/9708069.

\bibitem{Kupsch/Smolyanov:2000}
J.~Kupsch and O.~G. Smolyanov.
\newblock Hilbert norms for graded algebras.
\newblock {\em Proc. Amer. Math. Soc.}, 128:1647--1653, 2000.
\newblock arXiv:funct-an/9712005.

\bibitem{Ottesen:1995}
J.~T. Ottesen.
\newblock {\em {Infinite Dimensional Groups and Algebras in Quantum Physics}}.
\newblock Springer, Berlin, 1995.
\newblock {Lect. Notes Phys. Vol. m 27}.

\bibitem{Pressley/Segal:1986}
A.~Pressley and G.~Segal.
\newblock {\em Loop Groups}.
\newblock Clarendon Press, Oxford, 1986.

\bibitem{Rogers:1980}
A.~Rogers.
\newblock A global theory of supermanifolds.
\newblock {\em J. Math. Phys.}, 21:1352--1365, 1980.

\bibitem{Rogers:2007}
A.~Rogers.
\newblock {\em Supermanifolds: Theory and Applications}.
\newblock World Scientific, Singapore, 2007.
\newblock ISBN 9810212283.

\bibitem{Ruijsenaars:1977}
S.~N.~M. Ruijsenaars.
\newblock {On Bogoliubov transformations for systems of relativistic charged
  particles}.
\newblock {\em J. Math. Phys.}, 18:517--526, 1977.

\bibitem{Ruijsenaars:1978}
S.~N.~M. Ruijsenaars.
\newblock {On Bogoliubov transformations. II. The general case}.
\newblock {\em Ann. Phys. (N.Y.)}, 116:105--134, 1978.

\bibitem{Shale/Stinespring:1965}
D.~Shale and W.~F. Stinespring.
\newblock Spinor representations of infinite orthogonal groups.
\newblock {\em J. Math. Mech. (Indiana Univ. Math. J.)}, 14:315--322, 1965.

\bibitem{Valatin:1958}
J.~G. Valatin.
\newblock Comments on the theory of superconductivity.
\newblock {\em Nuovo Cimento}, 7(6):843--857, 1958.

\bibitem{Valatin:1961}
J.~G. Valatin.
\newblock Generalized {Hartree-Fock} method.
\newblock {\em Phys. Rev.}, 122(4):1012--1020, 1961.

\end{thebibliography}
\end{document}